\documentclass[twoside,11pt]{article} 
\usepackage{jmlr2e} 
\usepackage{color,graphicx}
\usepackage{eurosym}
\usepackage{amsmath}
\usepackage{amssymb}
\usepackage{verbatim}
\usepackage{url}
\usepackage{natbib}
\usepackage{subfig}
\usepackage{epsfig}
\usepackage{float}
\usepackage{algorithm}
\usepackage{algpseudocode}
\usepackage{pslatex}

\DeclareMathOperator*{\argmin}{argmin}

\DeclareMathOperator*{\Gam}{Gamma}
\DeclareMathOperator*{\Dir}{Dirichlet}

\def\bS{\mathbf{S}}
\def\tbS{\tilde{\bS}}
\def\tS{\tilde{S}}

\def\tT{\tilde{T}}

\def\cT{\mathcal{T}}
\def\cS{\mathcal{S}}
\newcommand{\M}{{\mathcal M}}
\newcommand{\Mu}{\mathcal{M}^{\cup}}

\newcommand{\dif}{\mathrm{d}}
\newcommand{\muu}{\mu^{\cup}}
\def\l({\left(}
\def\r){\right)}
\def\ind{1}

\jmlrheading{1}{2013}{1-26}{9/12; Revised 7/13}{}{Vinayak Rao, Yee Whye Teh}

\ShortHeadings{Fast MCMC Sampling for MJPs and Extensions}{Rao and Teh}
\firstpageno{1}

\begin{document}

\title{Fast MCMC Sampling for Markov Jump Processes and Extensions}
\author{\name Vinayak Rao\thanks{Corresponding author} \email vrao@gatsby.ucl.ac.uk \\
\addr Department of Statistical Science \\ Duke University \\ Durham, NC, 27708-0251, USA
\AND 
\name Yee Whye Teh \email y.w.teh@stats.ox.ac.uk \\
\addr Department of Statistics \\ 1 South Parks Road \\ Oxford OX1 3TG, UK \\ 
\AND
\name Accepted for publication in the Journal of Machine Learning Research (JMLR)
}

\editor{Christopher Meek}

\maketitle

\begin{abstract}%
Markov jump processes (or continuous-time Markov chains) are a simple and important class of continuous-time dynamical systems. 
In this paper, we tackle the problem of simulating from the posterior distribution over paths in these models, given partial and noisy observations.  Our approach is an auxiliary variable Gibbs sampler, and is based on the idea of 
\emph{uniformization}.  This sets up a Markov chain over paths by alternately sampling a finite set of
virtual jump times given the current path, and then sampling a new path given the set of extant and virtual jump times. The first step involves simulating
a piecewise-constant inhomogeneous Poisson process, while for the second, we use a standard hidden Markov model forward 
filtering-backward sampling algorithm. Our method is exact and does not involve approximations like time-discretization. We demonstrate how our sampler extends naturally to MJP-based models like Markov-modulated Poisson processes and continuous-time
Bayesian networks, and show significant computational benefits over state-of-the-art MCMC samplers for these models.
\end{abstract}%
%\vspace{.1in}\\

\section{Introduction}

The Markov jump process (MJP) extends the discrete-time Markov chain to continuous time, and forms a simple and popular class of continuous-time dynamical systems. 
In Bayesian modelling applications, the MJP is widely used as a prior distribution over the piecewise-constant evolution of the state of a system.
The Markov property of the MJP makes it both a realistic model for various physical and chemical systems, as well as a convenient approximation for more complex 
phenomena in biology, finance, queuing systems etc. 
In chemistry and biology, stochastic kinetic models use the state of an MJP to represent the sizes of various interacting \emph{species} 
\citep[e.g.,][]{gillespie1977exact, GolWilk11}. In queuing applications, the state may represent the number of pending jobs in a queue \citep{Breuer2003, tijms1986}, with the arrival
and processing of jobs treated as memoryless events. MJPs find wide application in genetics, %
for example, an MJP trajectory is sometimes used to represent a segmentation of a strand of genetic matter.
Here `time' represents position along the strand, with particular motifs occurring with different rates in different regions 
\citep{FearnSher2006}. %
MJPs are also widely used in finance, for example, \cite{Elliott06} use an MJP to model switches in the parameters that govern the dynamics of stock prices (the latter 
being modelled with a  L\'{e}vy\ process).

In the Bayesian setting, the challenge is to characterize the posterior distribution over MJP trajectories given noisy observations;
this typically cannot be performed analytically. Various sampling-based \citep{FearnSher2006, BoysWK08, El_hay_gibbssampling, FanShe08,Hobolth09} 
and deterministic \citep{Nodelman+al:UAI02,Nodelman+al:UAI05EP,Opper_varinf,Cohn_meanfield} approximations have been proposed in the literature,
but come with problems: they are often generic methods that do not exploit the structure of the MJP, and when they do, involve expensive computations like matrix exponentiation, 
matrix diagonalization or root-finding, or are biased, involving some form of time-discretization 
or independence assumptions. Moreover, these methods do not extend easily to more complicated likelihood functions which require specialized 
algorithms (for instance, the contribution of \cite{FearnSher2006} is to develop an exact sampler for Markov-modulated Poisson processes (MMPPs), where an MJP modulates 
the rate of a Poisson process).

In this work, an extension of \cite{RaoTeh2011a}, we describe a novel Markov chain Monte Carlo (MCMC) sampling algorithm for MJPs that avoids the need for the expensive 
computations described previously, and does not involve any form of approximation (i.e., our MCMC sampler converges to the true 
posterior). Importantly, our sampler is easily adapted to complicated extensions of MJPs such as MMPPs and continuous-time Bayesian networks (CTBNs) \citep{Nodelman+al:UAI02}, and is significantly 
more efficient than the
specialized samplers developed for these models. Like many existing methods, our sampler introduces auxiliary variables which simplify the structure of the MJP, 
using an idea called \emph{uniformization}. Importantly, unlike some existing methods which produce \emph{independent} posterior samples of these auxiliary 
variables,
our method samples these \emph{conditioned} on the current sample trajectory.  While the former approach depends on the observation process, and 
can be hard for complicated likelihood functions, ours results in a simple distribution over the auxiliary
variables that is independent of the observations. The observations are accounted for during a straightforward discrete-time forward-filtering backward-sampling 
step to resample a new trajectory.
The overall structure of our algorithm is that of an auxiliary variable Gibbs sampler, alternately resampling the auxiliary variables given the MJP trajectory, and
the trajectory given the auxiliary variables.

In Section \ref{sec:mjp} we briefly review Markov jump processes.  In Section \ref{sec:MJP_MCMC_UNIF} we introduce the idea of uniformization and describe our MCMC sampler for the simple
case of a discretely observed MJP.  In Section \ref{sec:mmpp}, we apply our sampler to the Markov-modulated Poisson process, while in Section \ref{ctbn}, we describe 
continuous-time Bayesian networks, and extend our algorithm to that setting.  In both sections, we report experiments comparing our algorithm to state-of-the-art sampling 
algorithms developed for these models. We end with a discussion in Section \ref{discussion}.

\section{Markov Jump Processes (MJPs)}\label{sec:mjp}
  
\begin{figure}[ht]
  \centering
\subfloat {
    \includegraphics[width=.486\textwidth]{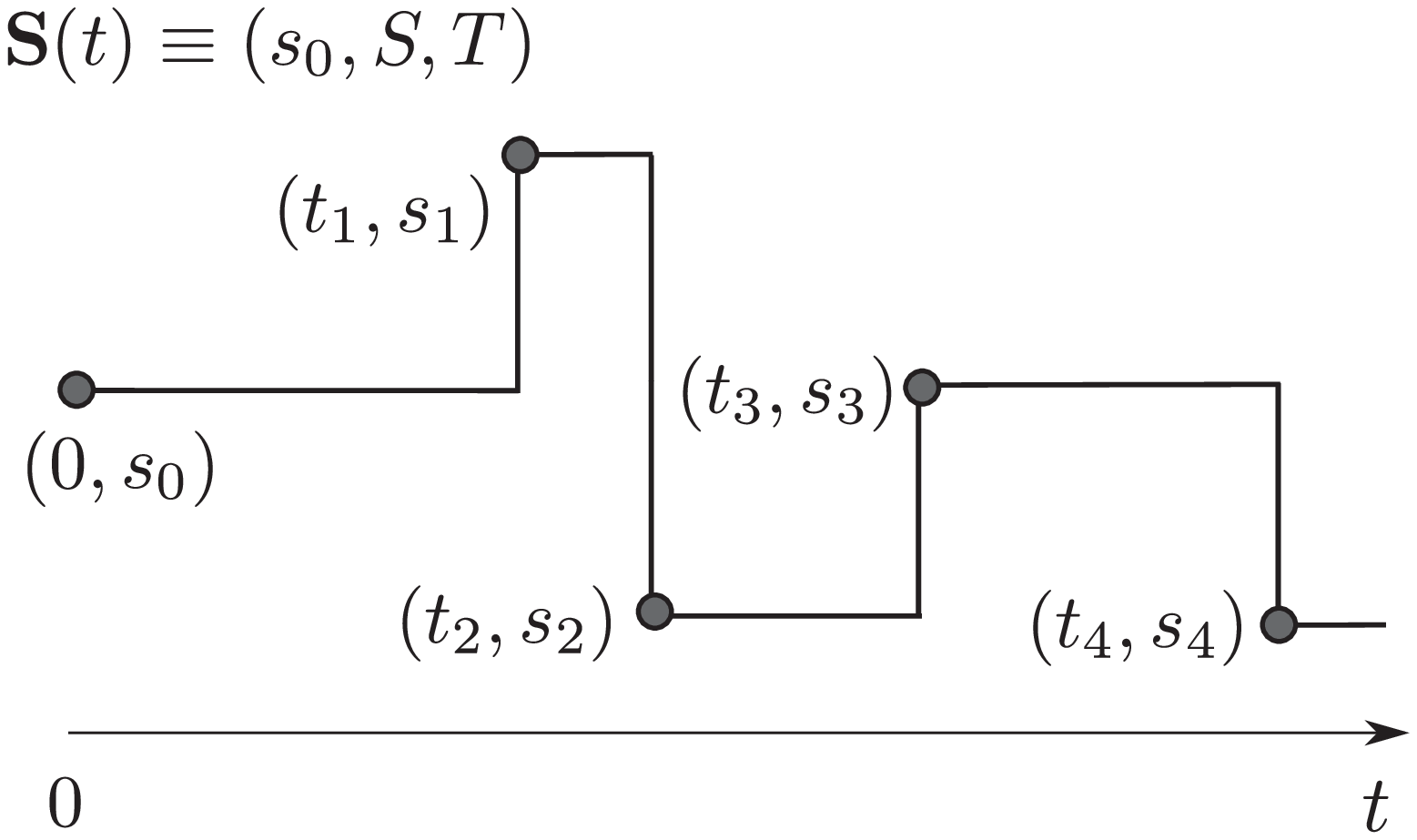} 
}
\subfloat {
    \includegraphics[width=.486\textwidth]{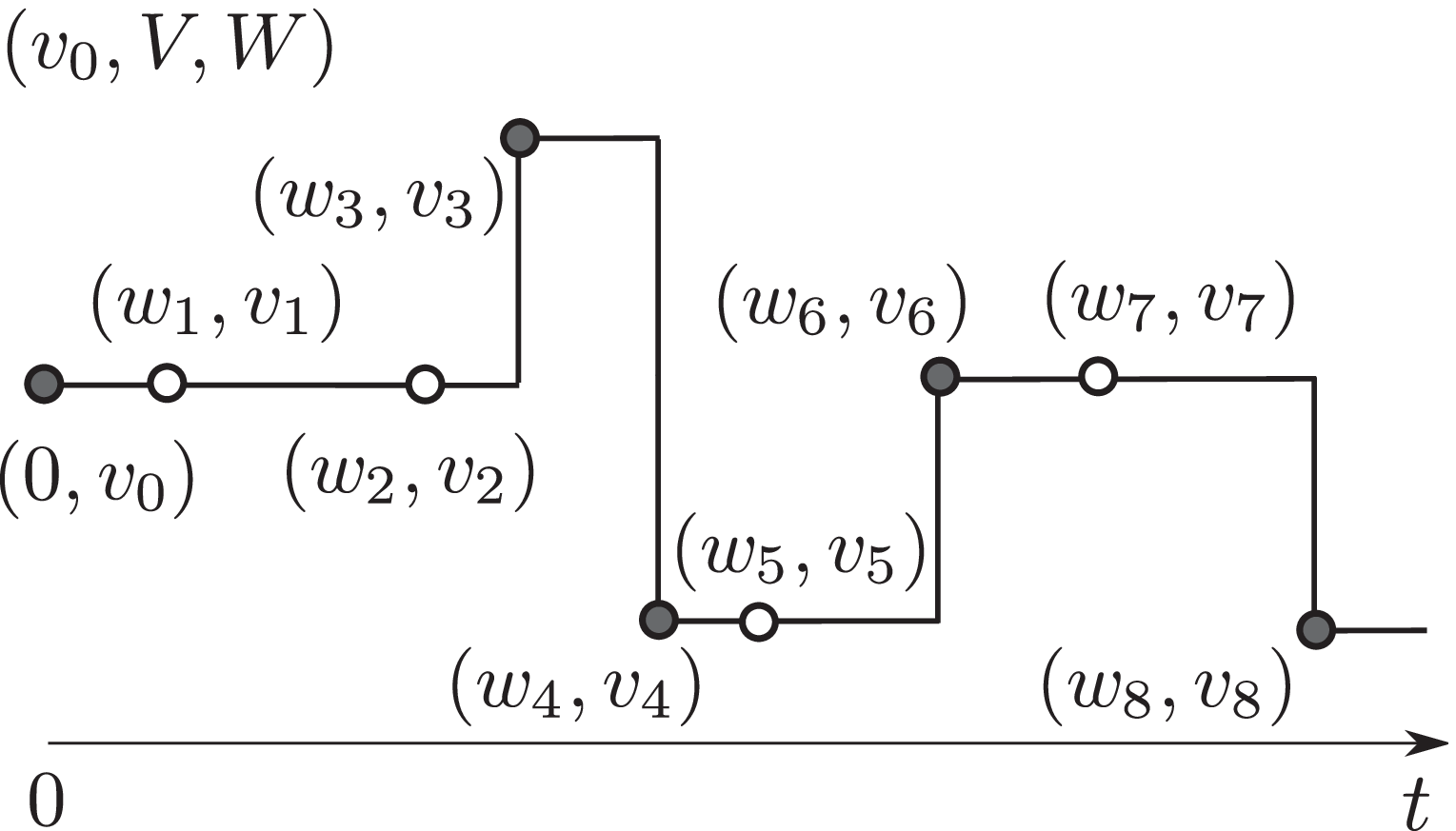} 
}
    \caption{(left) An MJP path $(s_0,S,T)$, (right) a uniformized representation $(v_0,V,W)$.}
    \label{fig:mjp_unif}
    \vspace*{-.5em}
  \end{figure}

  A Markov jump process $(\bS(t),\ t \in \mathbb{R_+})$ is a stochastic process with right-continuous, piecewise-constant paths \citep[see for example][]{Cinlar1975}. The paths 
themselves take values in some countable space $(\cS, \Sigma_{\cS})$, where $\Sigma_{\cS}$ is the discrete $\sigma$-algebra. As in typical applications, we assume $\cS$ is finite 
(say $\mathcal{S} = \{1,2,...N\}$). We also assume the 
process is homogeneous, implying (together with the Markov property) that for all times $t, t' \in \mathbb{R}_+$ and states $s,s' \in \mathcal{S}$,
\begin{align}
  p\left(\bS(t'+t) = s| \bS(t') = s', (\bS(u), u< t')\right) &=  [P_t]_{ss'}    \nonumber
\end{align}
for some stochastic matrix $P_t$ that depends only on $t$. The family of transition matrices $(P_t, \ t \ge 0)$ is defined by a matrix $A \in \mathbb{R}^{N \times N}$ called the 
\emph{rate matrix} or \emph{generator} of the MJP. 
$A$ is the time-derivative of $P_t$ at $t = 0$, with
\begin{align}
  P_t  &= \exp({At}) , \label{eq:matr_exp} \\
  p(\bS(t'+\dif t) = s| \bS(t') = s') &=  A_{ss'}\dif t  \hspace{.26in} (\textnormal{for } s \ne s') , \nonumber
\end{align}
where Equation \eqref{eq:matr_exp} is the matrix exponential. The off-diagonal elements of $A$ are nonnegative, and represent the rates of transiting from one state to 
another. Its diagonal entries are $A_s \equiv A_{ss} = -\sum_{s'\neq s} A_{s's}$ for each $s$, so that its columns sum to 0, with $-A_s = |A_s|$ characterizing
the total rate of leaving state $s$.

Consider a time interval $\cT \equiv [t_{start}, t_{end}]$, with the Borel $\sigma$-algebra $\Sigma_{\cT}$. Let $\pi_0$ be a density with respect to the counting
measure $\mu_{\cS}$ on $(\cS, \Sigma_{\cS})$; this defines the initial distribution over states at $t_{start}$. Then an MJP is described by the following 
generative process over paths on this interval, commonly called \emph{Gillespie's algorithm }\citep{gillespie1977exact}: 
\begin{algorithm}[H]
\caption{Gillespie's algorithm to sample an MJP path on the interval $[t_{start}, t_{end}]$}\label{alg:gillesp}
\begin{tabular}{p{1.4cm}p{12.2cm}}
\textbf{Input:}  & The rate matrix $A$ and the initial distribution over states $\pi_0$.\\
\textbf{Output:} & An MJP trajectory $\bS(t) \equiv (s_0, S,T)$.\\
\vspace{-.1in}
\line(1,0){342}
\end{tabular}
\begin{algorithmic}[1]
\State Assign the MJP a state $s_0 \sim \pi_0$. Set $t_0=t_{start}$ and $i=0$.
\Loop
\State Draw $z \sim \exp(|A_{s_i}|)$.
\State \textbf{If} {$t_i+z > t_{end}$} \textbf{then} \textbf{return} {$(s_0,\ldots,s_i,t_1,\ldots,t_i)$} and \textbf{stop}.
\State Increment $i$ and let $t_{i} = t_{i-1} + z$.
\State The MJP jumps to a new state $s_i=s$ at time $t_i$, for an $s\neq s_{i-1}$,  \\ \hspace*{3em} with probability proportional to $A_{ss_{i-1}}$.
\EndLoop
\end{algorithmic}
\end{algorithm}

If all event rates are finite, an MJP trajectory will almost surely have only a finite number of jumps. Let there be $n$ jumps, and let these occur at the ordered times 
$(t_1, \cdots, t_n)$. Define $T \equiv (t_1, \cdots, t_n)$, and let $S \equiv (s_1, \cdots, s_n)$ be the corresponding sequence of states, where $s_i = \bS(t_i)$. 
The triplet $(s_0,S,T)$ completely characterizes the MJP trajectory over $\cT$ (Figure
\ref{fig:mjp_unif} (left)).

From Gillespie's algorithm, we see that sampling an MJP trajectory involves sequentially sampling $n+1$ waiting times from exponential densities with
one of $N$ rates, and $n$ new states from one of $N$ discrete distributions,  each depending on the previous state. The $i$th waiting time equals $(t_i - t_{i-1})$ and is drawn from an exponential
with rate $|A_{s_{i-1}}|$, while the probability the $i$th state
equals $s_i$ is $A_{s_is_{i-1}}/|A_{s_{i-1}}|$. The last waiting time can take any value greater than $t_{end} - t_n$. Thus, under an MJP, 
a random element $(s_0, S,T)$ %
has density
\begin{align}
  p(s_0, S,T)  &=  \pi_0(s_0) \left(\prod_{i=1}^{n}  \textstyle|A_{s_{i-1}}| e^{-|A_{s_{i-1}}|(t_{i} - t_{i-1})} \frac{A_{s_is_{i-1}}}{|A_{s_{i-1}}|} \right) 
               \cdot e^{-|A_{s_{n}}|(t_{end} - t_{n})} \nonumber \\   %\label{eq:path_prob1}
        &= \pi_0(s_0) \left(\prod_{i=1}^{{n}}  A_{s_is_{i-1}}\right)  \exp\left( -\int^{t_{end}}_{t_{start}} |A_{\bS(t)}| \dif t\right).  \label{eq:path_prob2} 
\end{align}
To be precise, we must state the base measure with respect to which the density above is defined. The reader might wish to skip these 
details \citep[and for more details, we recommend][]{DalVer2008a}.  Let $\mu_{\cT}$ be Lebesgue measure on $\cT$.
Recalling that the state space of the MJP is $\cS$, we can view $(S,T)$ as a sequence of elements in the product space 
$\M \equiv \cS \times \cT$. Let $\Sigma_{\M}$ and $\mu_{\M} = \mu_{\cS} \times \mu_{\cT}$ be the 
corresponding product $\sigma$-algebra and product measure. Define $\M^n$ as the $n$-fold product space with the usual
product $\sigma$-algebra $\Sigma^n_{\M}$ and product measure $\mu_{\M}^n$.  Now  let $\Mu \equiv \bigcup_{i=0}^{\infty} \M^i$ be a union space, elements of which
represent finite length pure-jump paths\footnote{Define $\M^0$ as a point satisfying $\M^0 \times \M = \M \times \M^0 = \M$ \citep{DalVer2008a}.}.
Let $\Sigma_{\M}^\cup$ be the corresponding union $\sigma$-algebra, where each measurable set $B \in \Sigma_{\M}^\cup$ can be expressed as  $B =\cup_{i=0}^\infty B^i$ with $B^i =B\cap \M^i \in \Sigma^i_{\M}$.
Assign this space the measure $\mu^{\cup}_{\M}$ defined as:
%
%
%\vspace{-.2in}
\begin{align}
  \mu^{\cup}_{\M}(B) &= \sum_{i=0}^{\infty} \mu^i_\M(B^i) . \nonumber
\vspace{-.1in}
\end{align}
Then, any element $(s_0, S,T) \in \cS \times \Mu$ sampled from Gillespie's algorithm has density w.r.t.\ $\mu_{\cS} \times \mu^{\cup}_{\M}$ given by Equation \eqref{eq:path_prob2}.

\section{MCMC Inference via Uniformization }  \label{sec:MJP_MCMC_UNIF}

  In this paper, we are concerned with the problem of sampling MJP paths over the interval $\cT \equiv [t_{start}, t_{end}]$ given noisy observations of the state of the MJP.  
In the simplest case, we observe the process at the boundaries $t_{start}$ and $t_{end}$. More generally, we are given the initial distribution 
over states $\pi_0$ as well as a set of $O$ noisy observations $X = \{X_{t^o_1}, ... X_{t^o_{O}}\}$ at times $T^o = \{t^o_1,\dots,t^o_O\}$ with likelihoods $p(X_{t^o_i}|\bS(t^o_i))$, and we wish to sample from the 
posterior $p(s_0, S,T|X)$. Here we have implicitly assumed that the observation times $T^o$ are fixed.  Sometimes the observation times themselves
can depend on the state of the MJP, resulting effectively in \emph{continuous-time} observations. This is the case for the Markov-modulated Poisson process and 
CTBNs. As we will show later, %
our method handles these cases quite naturally as well. %

A simple approach to inference is to discretize time and work with the resulting approximation. The time-discretized MJP corresponds to the familiar discrete-time Markov
chain, and its Markov structure can be exploited to construct dynamic programming algorithms like the forward-filtering backward-sampling (FFBS) algorithm 
(\cite{fru:StateSpace, Carter96}; see also Appendix \ref{sec:ffbs}) to sample posterior
trajectories efficiently. However, time-discretization introduces a bias into our inferences, as the system can change state only at a fixed set of times, and as the maximum
number of state changes is limited to a finite number. To control this bias, one needs to discretize time at a fine granularity, resulting in long Markov chains, and expensive computations.

Recently, there has been growing interest in constructing \emph{exact} MCMC samplers for MJPs without any approximations such as time-discretization. We review these in Section
\ref{sec:related}. One class of methods exploits the fact an MJP can be exactly represented by a discrete-time Markov chain on a \emph{random} time-discretization. Unlike 
discretization on a regular grid, a random grid can be quite coarse without introducing any bias. Given this discretization, we can use the FFBS algorithm to perform efficient sampling. However, we do not observe the random discretization, and thus also need to sample this 
from its posterior distribution. This posterior now depends on the likelihood process, and a number of algorithms attempt to solve this problem for specific
observation processes.
Our approach is to resample the discretization conditioned on the system trajectory. As we will see this is \emph{independent} of the likelihood process, resulting in a simple, flexible and 
efficient MCMC sampler.

\subsection{Uniformization}
  We first introduce the idea of \emph{uniformization} \citep{Jen1953, Cinlar1975, Hobolth09}, which forms the basis of our sampling algorithm. For an MJP with 
rate-matrix $A$, choose some $\Omega \ge \max_s{|A_s|}$. Let $W=(w_1,\ldots,w_{|W|})$ be an ordered set of times on the interval $[t_{start}, t_{end} ]$ drawn from a homogeneous 
Poisson process with intensity $\Omega$. %
$W$ constitutes a random discretization of the time-interval $[t_{start},t_{end}]$. 

Next, letting $I$ be the identity matrix, observe that $B = \left(I + \frac{1}{\Omega}A \right)$ is a stochastic matrix (it has nonnegative elements, and its columns sum
to one). 
Run a discrete-time Markov chain with initial distribution $\pi_0$ and transition matrix $B$ on the times in $W$; this is a Markov chain \emph{subordinated} 
to the Poisson process $W$.  The Markov chain will assign a set of states $(v_0,V)$; $v_0$ at the initial time $t_{start}$, and $V=(v_1,\ldots,v_{|V|})$ at the 
discretization times $W$ (so that $|V| = |W|$).  In particular, $v_0$ 
is drawn from $\pi_0$, while $v_i$ is 
drawn from the probability vector given by the ${v_{i-1}}${th} column of $B$. 
Just as $(s_0, S,T)$ characterizes an MJP path, $(v_0, V,W)$ also characterizes a sample path of some piecewise-constant, right-continuous stochastic process on $[t_{start},t_{end}]$. 
Observe that the matrix $B$ allows self-transitions, so that unlike $S$, $V$ can jump from a state back to the same state.  
We treat these as \emph{virtual} jumps, and regard $(v_0,V,W)$ as a redundant representation of a pure-jump process that always jumps to a new state
(see Figure \ref{fig:mjp_unif} (right)).
The virtual jumps provide a mechanism to `thin' the set $W$, thereby rejecting some of its events. This corrects for the fact that since the Poisson rate $\Omega$ dominates the leaving rates of all 
states of the MJP, $W$ will on average contain more events than there are jumps in the MJP path. 
 As the parameter 
$\Omega$ increases, the number of events in $W$ increases; at the same time the diagonal entries of $B$ start to dominate,
so that the number of self-transitions (thinned events) also increases.   The next proposition shows that these two effects exactly compensate each other, 
so that the process characterized by $(v_0,V,W)$ is precisely the desired MJP.

\begin{prop}{\citep{Jen1953}}
 For any $\Omega \ge \max_s{|A_s|}$, $(s_0,S,T)$ and $(v_0,V,W)$ define the same Markov jump process $\bS(t)$.\label{prop:unif}
\end{prop}
\begin{proof}
We follow \cite{Hobolth09}. From Equation \eqref{eq:matr_exp}, 
the marginal distribution of the MJP at time $t$ is given by
\begin{align}
  \pi_t &= \exp(At) \pi_0 \nonumber \\
        &= \exp\l( \Omega{(B - I)t}\r) \pi_0 \nonumber \\
        &= \exp\l( -\Omega t\r) \exp\l( \Omega t{B}\r) \pi_0 \nonumber \\
        &=  \sum_{n=0}^{\infty} \l( \l( \exp\l( -\Omega t\r) \frac{\l( \Omega t\r)^n}{n!}\r) \l( B^n \pi_0\r) \nonumber \r) .
\end{align}

The first term in the summation is the probability that a rate $\Omega$ Poisson produces $n$ events in an interval of length $t$, i.e., that $|W| = n$. The 
second term gives the marginal distribution over states for a 
discrete-time Markov chain after $n$ steps, given that the initial state is drawn from $\pi_0$, 
and subsequent states are assigned according to a transition matrix $B$. Summing over $n$, we obtain the marginal distribution over states at time $t$.
Since the transition kernels induced by the uniformization procedure agree with those of the Markov jump process ($\exp(At)$) for all $t$, and since the two 
processes also share the same initial distribution of states, $\pi_0$, all finite dimensional distributions agree. Following Kolmogorov's extension 
theorem \citep{kallenberg_FMP}, both define versions of the same stochastic process.

\end{proof}

A more direct but cumbersome approach is note that $(v_0,V,W)$ is also an element of the space $\cS \times \M^{\cup}$. 
We can then write down its density $p(v_0,V,W)$ w.r.t.\ $\mu_{\cS} \times \mu^{\cup}$, and 
show that marginalizing out the number and locations of self-transitions recovers Equation \eqref{eq:path_prob2}.
While we do not do this, we will derive the density $p(v_0,V,W)$ for later use.  As in Section \ref{sec:mjp}, let
$\cT^{\cup}$ and $\cS^{\cup}$ denote the measure spaces consisting of finite sequences of times and states respectively, and let $\muu_{\cT}$ and $\muu_{\cS}$ be the corresponding base measures.
The Poisson realization $W$ is determined by waiting times sampled from a rate $\Omega$ exponential distribution,
so that following Equation \eqref{eq:path_prob2}, $W$ has density w.r.t.\ $\mu^{\cup}_{\cT}$ given by
\begin{align}
  p(W) =  {\Omega^{|W|}}{e^{-\Omega(t_{end}-t_{start})}} . \label{eq:poiss_dens}
\end{align}

Similarly, from the construction of the Markov chain, it follows that the state assignment $(v_0, V)$ has probability density w.r.t.\ 
$\mu_{\cS} \times \mu^{\cup}_{\cS}$ given by
\begin{align}
  p(v_0, V|W) %
         & = \pi_0(v_0) \prod_{i=1}^{|V|} \left(1+ \frac{A_{v_i}}{\Omega}\right)^{\ind(v_i = v_{i-1})}  \nonumber
             \left(\frac{A_{v_{i} v_{i-1}}}{\Omega}\right)^{\ind(v_i \neq v_{i-1})} . \label{eq:mc_dens}
\end{align}
Since under uniformization $|V| = |W|$, it follows that %
\begin{align}
\muu_{\cS}(\dif V) \times \muu_{\cT}(\dif W) &=   \mu_{\cS}^{|V|}(\dif V) \times \mu^{|W|}_{\cT}(\dif W) \nonumber \\
  & = (\mu_{\cT} \times \mu_{\cS})^{|V|}(\dif(V,W)) \nonumber \\
  &= \muu_\M(\dif(V,W)).
\end{align}
Thus, from Equations \eqref{eq:poiss_dens} and \eqref{eq:mc_dens}, $(v_0, V,W)$ has density w.r.t.\ $\mu_{\cS} \times \muu_\M$ given by
\begin{align}
 \hspace{-.22in} p(v_0,V,W)={e^{-\Omega(t_{end}-t_{start})}}
       \pi_0(v_0)  \prod_{i=1}^{|V|} \left(\Omega + {A_{v_i}}\right)^{\ind(v_i = v_{i-1})} \left({A_{v_{i} v_{i-1}}}\right)^{\ind(v_i \neq v_{i-1})} . \label{eq:p_v_w}
\end{align}

\subsection{The MCMC Algorithm}  \label{sec:MJP_MCMC}
We adapt the uniformization scheme described above to construct an auxiliary variable Gibbs sampler.  Recall that the only difference between $(s_0,S,T)$ and $(v_0,V,W)$ is the presence of an 
auxiliary set of virtual jumps in the latter. Call the virtual jump times $U_{\cT}$; associated with $U_{\cT}$ is a sequence of states $U_{\cS}$.
$U_{\cS}$ is  uniquely determined by $(s_0,S,T)$ and $U_{\cT}$ (see Figure \ref{fig:mjp_unif}(right)), and we say this configuration is
\emph{compatible}. Let $U = (U_{\cS}, U_{\cT})$, and observe that for compatible values of $U_\cS$, $(s_0, S, T, U)$ and $(v_0, V,W)$
represent the same point in ${\cS} \times \M^{\cup}$.

  \begin{figure}[ht]
  \begin{minipage}[b]{0.32\linewidth}
  \centering
    \includegraphics[width=0.9\textwidth]{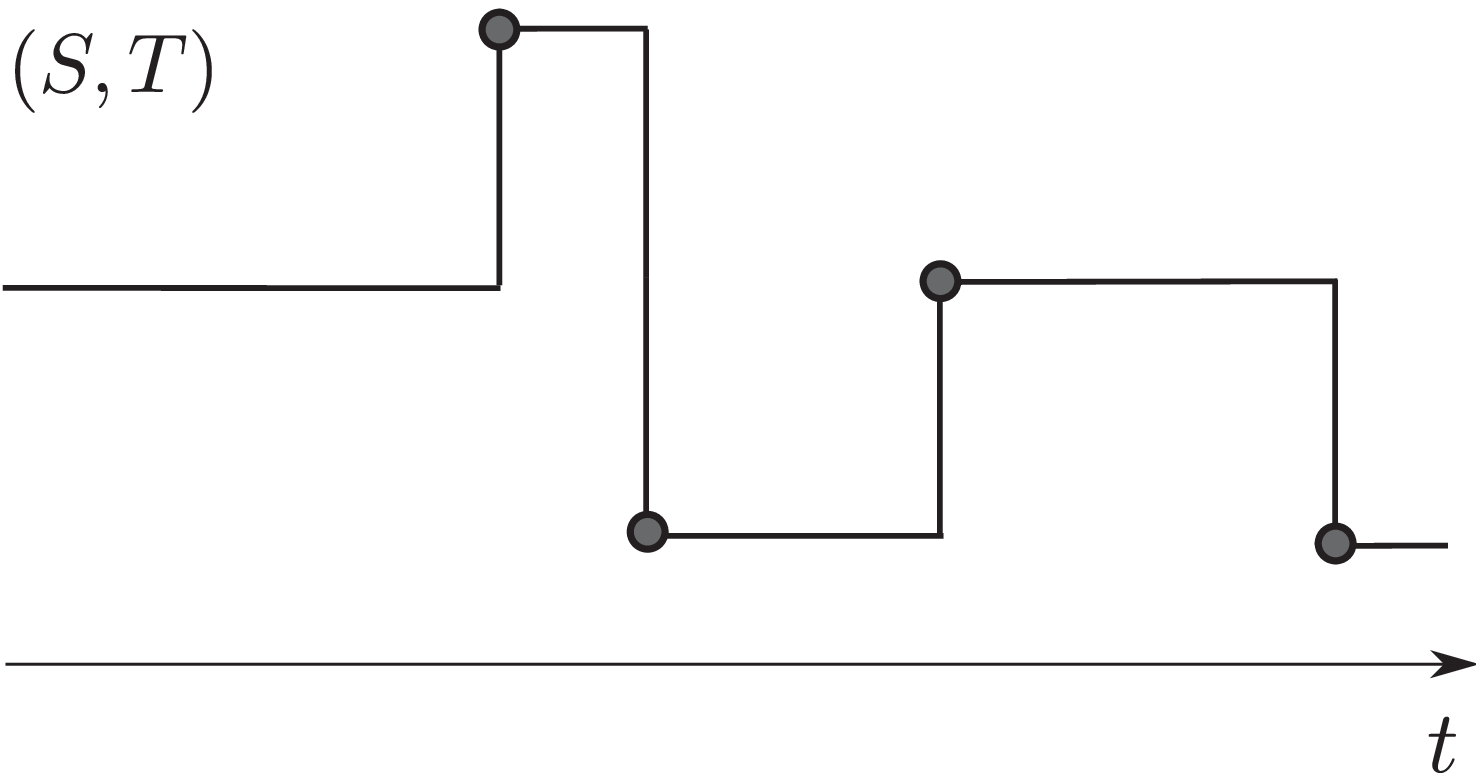} 
  \end{minipage}
  \begin{minipage}[b]{0.32\linewidth}
  \centering
    \includegraphics[width=0.9\textwidth]{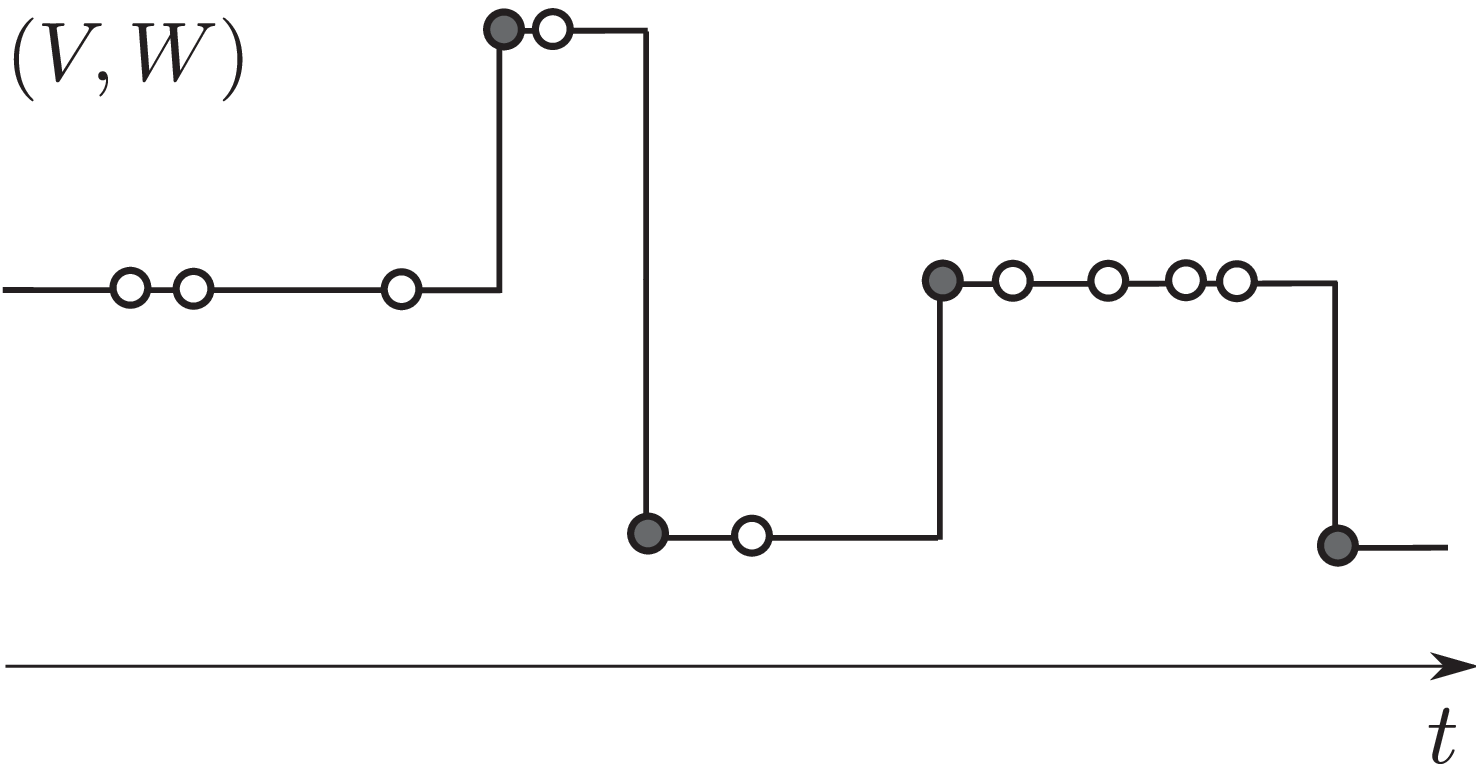} 
  \end{minipage}
  \begin{minipage}[b]{0.32\linewidth}
  \centering
    \includegraphics[width=0.9\textwidth]{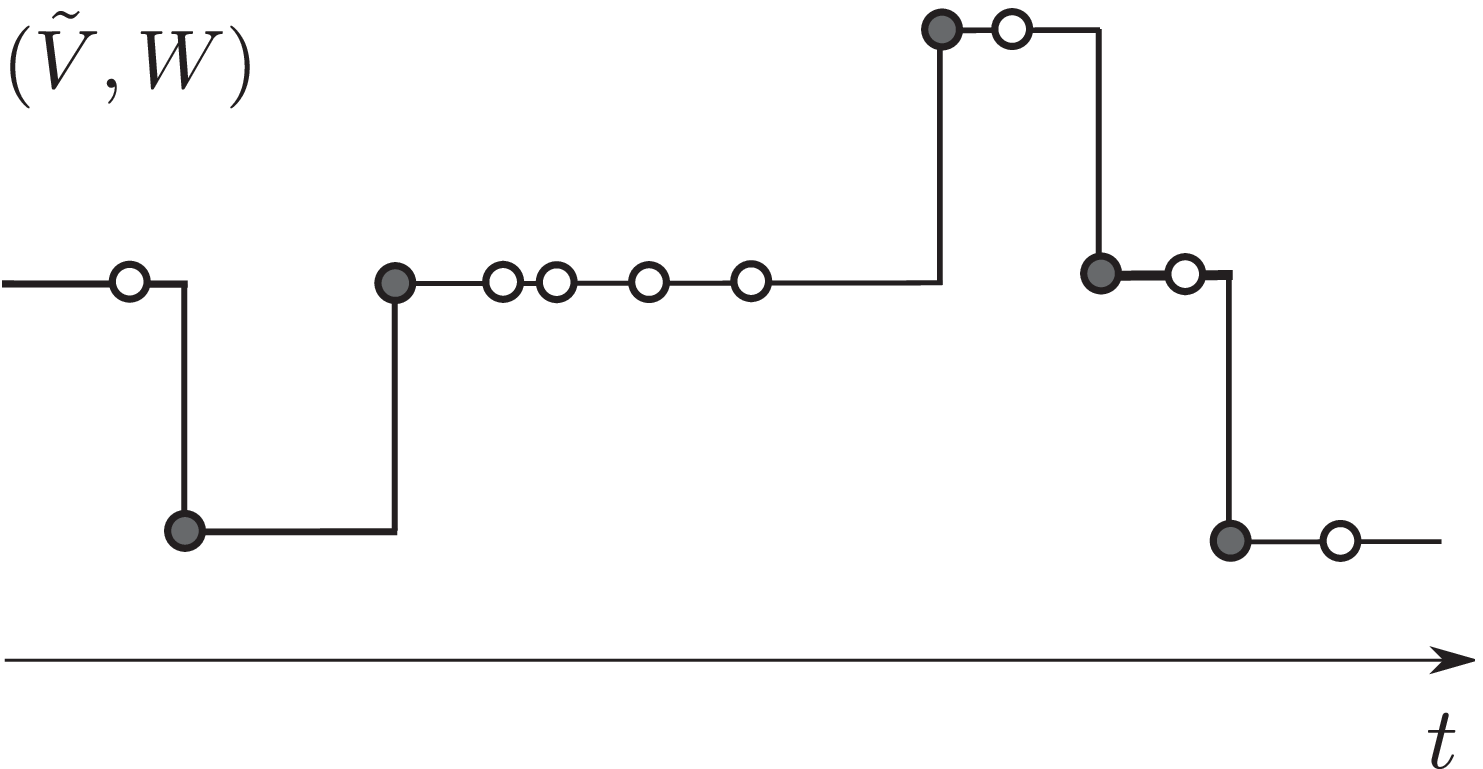} 
  \end{minipage}
    \caption{Uniformization-based Gibbs sampler: starting with an MJP trajectory (left), resample the thinned events (middle) and then resample the trajectory given all Poisson events (right). Discard the thinned events and repeat.}
  \label{fig:mcmc_fig}
  \end{figure}

%Each iteration of the MCMC algorithm then proceeds as follows.  
Given an MJP trajectory $(s_0,S,T)$ (Figure \ref{fig:mcmc_fig} (left)), we proceed by first sampling the set of virtual jumps $U_\cT$ given $(s_0,S,T)$,
 as a result recovering the uniformized characterization $(s_0,V,W)$ (Figure \ref{fig:mcmc_fig} (middle)). 
 This corresponds to a random discretization of $[t_{start}, t_{end}]$ at times $W$. %
We now discard the state sequence $V$, and perform a simple HMM forward-filtering backward-sampling step to resample a new state sequence $\tilde{V}$.
Finally, dropping the virtual jumps in $(s_0,\tilde{V},W)$ gives a new MJP path $(s_0,\tilde{S},\tilde{T})$.
Figure \ref{fig:mcmc_fig} describes an iteration of the MCMC algorithm.

The next proposition shows that conditioned on $(s_0,S,T)$, the virtual jump times $U_{\cT}$ are distributed as an \emph{inhomogeneous} Poisson process with 
intensity $R(t) =  \Omega + A_{\bS(t)}$ (we remind the reader that $A$ has a negative diagonal, so that $R(t) \le \Omega$). This intensity is 
piecewise-constant, taking the value $r_{i} = \Omega + A_{s_{i}}$ on the interval $[t_{i},t_{i+1})$ (with $t_0=t_{start}$ and $t_{n+1}=t_{end}$), so it is easy to sample $U_{\cT}$ and thus $U$.  %

\begin{prop}
For any $\Omega \ge \max_s{(|A_s|)}$, both $(s_0,S,T,U)$ and $(v_0,V,W)$ have the same density w.r.t.\ $\mu_{\cS} \times \muu_\M$.  
In other words, the Markov jump process $(s_0,S,T)$ along with virtual jumps $U$ drawn from the inhomogeneous Poisson process as above is equivalent to the times $W$ being drawn from a Poisson process with rate $\Omega$, followed by the states $(v_0,V)$ being drawn from the subordinated Markov chain.  
\end{prop}

\begin{proof}
Let $n=|T|$ be the number of jumps in the current MJP trajectory.  Define $|U_{i}|$ as the number of auxiliary times in interval $(t_{i},t_{i+1})$.
Then, $|U_\cT| = \sum_{i=0}^{n} |U_i|$.  If $U_\cT$ is sampled from a piecewise-constant inhomogeneous Poisson process, its density is the product of the 
densities of a sequence of homogeneous Poisson processes, and from Equation \eqref{eq:poiss_dens} is % it has density
\begin{align}
  p(U_\cT|s_0,S,T)  %
            &=  \left(\prod_{i=0}^{n} {(\Omega + A_{s_i})^{|U_i|}}\right) \exp\left(-\int^{t_{end}}_{t_{start}}(\Omega + A_{\bS(t)})dt\right) \hspace*{-.5em}\label{eq:aux}
\end{align}
w.r.t.\  $\muu_\cT$.  Having realized the times $U_\cT$, the associated states $U_\cS$ are determined too (elements of $U_\cS$ in the interval $(t_{i-1},t_i)$ equal 
$s_{i-1}$). Thus $U=(U_\cS,U_\cT)$ given $(s_0,S,T)$ has the same density as Equation \eqref{eq:aux}, but now w.r.t.\ $\muu_\M$, {and now restricted to
elements of $\M^{\cup}$ where $U_\cS$ is compatible with $(S,T,U_\cT)$}. 
Multiplying Equations \eqref{eq:path_prob2} and \eqref{eq:aux}, we see that $(s_0, S, T, U)$ has density
\begin{align}
  p(s_0,S,T,U)
%  &= \frac{\Omega^{|U|+n}}{e^{\Omega(t_{end}-t_{start})}} \cdot
%                      \pi_0(s_0)\prod_{i=0}^n {\left(1+ \frac{A_{s_i}}{\Omega}\right)^{|U_i|}} \prod_{i=1}^n \frac{A_{s_{i} s_{i-1}}}{\Omega} \nonumber \\
  &= e^{-\Omega(t_{end}-t_{start})} 
                      \pi_0(s_0)\prod_{i=0}^n {\left(\Omega+A_{s_i}\right)^{|U_i|}} \prod_{i=1}^n A_{s_{i} s_{i-1}} \nonumber
\end{align}
w.r.t.\ $\mu_{\cS} \times \muu_\M \times \muu_\M$.
However, by definition, 
\begin{align}
\muu_\M(\dif(S,T)) \times \muu_\M(\dif U) & =\mu^{|T|}_\M(\dif(S,T)) \times \mu^{|U|}_\M(\dif U)\nonumber\\
  & = \mu^{|T|+|U|}_\M(\dif(S,T,U))\nonumber\\
  & = \muu_\M(\dif(S,T,U)) . \nonumber
\end{align}
Comparing with Equation \eqref{eq:p_v_w}, and noting that $|U_i|$ is the number of self-transitions in interval $(t_{i-1},t_i)$, we see both are
equal whenever $U_\cS$ is compatible with $(s_0,S,T,U_\cT)$.
The probability density at any incompatible setting of $U_\cS$ is zero, giving us the desired result.
\end{proof}

We can now incorporate the likelihoods coming from the observations $X$.  Firstly, note that by assumption, $X$ depends only on the MJP trajectory $(s_0,S,T)$ and not on the auxiliary jumps $U$. Thus, the conditional distribution of $U_\cT$ given $(s_0,S,T,X)$ is still the inhomogeneous Poisson process given above.  Let $X_{[w_i,w_{i+1})}$ represent the observations in the interval 
$[w_i,w_{i+1})$ (taking $w_{|W|+1} = t_{end}$). Throughout this interval, the MJP is in state $v_i$, giving a likelihood term:
\begin{align}
  L_i(v_i) &= p(X_{[w_i, w_{i+1})}|\bS(t)=v_i \text{ for } t \in [w_i,w_{i+1})) .  \label{op_lik_gen} \\
\intertext{For the case of noisy observations of the MJP state at a discrete set of times $T^o$, this simplifies to}
  L_i(v_i) &= \prod_{j:t^o_j\in[w_i,w_{i+1})} p(X_{t^o_j}|\bS(t^o_j)=v_i) .  \nonumber
\end{align}

Conditioned on the times $W$, $(s_0,V)$ is a Markov chain with initial distribution $\pi_0$, transition matrix $B$ and likelihoods given by Equation 
\eqref{op_lik_gen}. We can efficiently resample $(s_0,V)$ using the standard forward filtering-backward sampling (FFBS) algorithm. We provide a description of this algorithm
in Appendix \ref{sec:ffbs}.
This cost of such a resampling step is $O(N^2|V|)$, quadratic in the number of states and linear in the length of the chain.  Further, any structure in $A$ (e.g., sparsity) is inherited by $B$ and can be exploited easily.
Let $(\tilde{s}_0,\tilde{V})$ be the new state sequence.  Then $(\tilde{s}_0,\tilde{V},W)$ will correspond to a new MJP path $\tilde{\bS}(t) \equiv (\tilde{s}_0,\tilde{S},\tilde{T})$,
obtained by discarding virtual jumps from $(\tilde{V},W)$. Effectively, given an MJP path, an iteration of our algorithm corresponds to introducing thinned
events, relabelling the thinned and actual transitions using FFBS, and then discarding the new thinned events to obtain a new MJP.
We summarize this in Algorithm \ref{alg:mcmc_smplr}.

\begin{algorithm}[H]
\caption{Blocked Gibbs sampler for an MJP on the interval $[t_{start},t_{end}]$}\label{alg:mcmc_smplr}
\label{alg:blk_gibbs}
\begin{tabular}{p{1.4cm}p{10.5cm}}
\textbf{Input:}  & A set of observations $X$, and parameters $A$ (the rate matrix),
                 $\pi_0$ (the initial distribution over states) and $\Omega > \max_s(|A_s|)$.\\
               &The previous MJP path, $\bS(t) \equiv (s_0, S, T)$.\\%

\textbf{Output:} &A new MJP trajectory $\tbS(t) \equiv (\tilde{s}_0, \tilde{S},\tilde{T})$.\\
\hline 
\end{tabular}
\begin{algorithmic}[1]
\State Sample $U_\cT \subset [t_{start},t_{end}]$ from a Poisson process with piecewise-constant rate $$R(t)=(\Omega + A_{\bS(t)}).$$ Define $W = T \cup U_\cT$ (in increasing order).
\State Sample a path $(\tilde{s}_0,\tilde{V})$ from a discrete-time Markov chain with 
       $1 + |W|$ steps using the FFBS algorithm. The transition matrix of the Markov chain is 
$B = \left(I + \frac{A}{\Omega}\right)$ while the initial distribution over states is $\pi_0$.
     The likelihood of state $s$ at step $i$ is $$L_i(s) = p(X_{[w_i,w_{i+1})} | \bS(t) = s \text{ for } t\in[w_i,w_{i+1})).$$
\State Let $\tilde{T}$ be the set of times in $W$ when the Markov chain changes state. 
       Define $\tilde{S}$ as the corresponding set of state values. 
\textbf{Return} $(\tilde{s}_0,\tilde{S},\tilde{T})$.
\end{algorithmic}
\end{algorithm}

\begin{prop}  \label{prop:ergdcty}
The auxiliary variable Gibbs sampler described above has the posterior distribution $p(s_0,S,T|X)$ as its stationary distribution. 
Moreover, if $\Omega > \max_s|A_s|$, the resulting Markov chain is irreducible.
\end{prop}
\begin{proof}
  The first statement follows since the algorithm simply introduces auxiliary variables $U$, and then conditionally samples $V$ given $X$ and $W$.  
For the second, note that if $\Omega > \max_s(|A_s|)$, then the intensity of the subordinating Poisson process is strictly positive. 
Thus, there is positive probability density of sampling appropriate auxiliary jump times $U$ and moving from any MJP trajectory to any other.
\end{proof}

Note that it is essential for $\Omega$ to be strictly greater than $\max_s|A_s|$; equality is not sufficient for irreducibility.  For example, if all diagonal elements of $A$ are equal to $\Omega$, then the 
Poisson process for $U_\cT$ will have intensity 0, so that the set of jump times $T$ will never increase.  

Since FFBS returns a new state sequence $\tilde{V}$ that is independent of $V$ given $W$, the 
only dependence between successive MCMC samples %
arises because the new candidate jump times include the old jump times i.e.,  $T \subset W$.
This means that the new MJP trajectory has non-zero probability of making a jump at a same time as the old trajectory.
Increasing $\Omega$ introduces more virtual jumps, %
and as $T$ becomes a smaller subset of $W$, we get faster mixing. 
Of course, increasing $\Omega$ makes the HMM chain grow longer, leading to a linear increase in the computational cost per 
 iteration.  Thus the parameter $\Omega$ allows a trade-off between mixing rate and computational cost. We will study this trade-off in Section \ref{mjp:expt}. 
In all other experiments, we set $\Omega =  \max_s({2|A_s|})$ as we find this works quite well, with the samplers typically converging after fewer than 5 iterations.

\subsection{Previous Posterior Sampling Schemes} \label{sec:related}
A simple approach when the MJP state is observed at the ends of an interval is rejection sampling: sample paths given the observed start state via Gillespie's algorithm,
 and reject those that do not end
in the observed end  state \citep{Niels2002}. We can extend this to noisy observations by importance sampling or particle filtering \citep{FanShe08}. Recently,
\cite{GolWilk11} have applied particle MCMC methods to correct the bias introduced by standard particle filtering methods. However, these methods are efficient only in situations
where the data exerts a relatively weak influence on the unobserved trajectory (compared to the prior): a large state-space or an unlikely end state can result in a large number of 
rejections or small effective sample sizes  \citep{Hobolth09}.

A second approach, more specific to MJPs, integrates out the infinitely many paths of the MJP in between observations using matrix exponentiation 
(Equation \eqref{eq:matr_exp}), and uses forward-backward dynamic programming to sum over the states at the finitely many observation times (see \cite{Hobolth09} for a review).  Unfortunately, matrix exponentiation
is an expensive operation that scales as $O(N^3)$, cubically in the number of states. Moreover, the matrix resulting from matrix exponentiation is dense and any structure (like sparsity), in the rate matrix $A$ cannot be exploited.

A third approach is, like ours, based on the idea of uniformization \citep{Hobolth09}. 
This proceeds by producing independent posterior samples of the Poisson events $W$ in the interval between observations, and then (like our sampler) running a discrete-time Markov chain
on this set of times to sample a new trajectory. However, sampling from the posterior distribution over $W$ is not easy, 
depending crucially on the observation process, and usually requires a random number of $O(N^3)$ matrix multiplications (as the sampler iterates over the possible number of Poisson events). 
By contrast, instead of producing {independent} samples, ours is an {MCMC} algorithm.
At the price of producing dependent samples, our method scales as $O(N^2)$ given a random discretization of time, 
does not require matrix exponentiations, easily exploits 
structure in the rate matrix and naturally extends to various extensions of MJPs. Moreover, we demonstrate that our sampler mixes very rapidly.
We point out here that as the number of states $N$ increases, if the transition rates $A_{ss'}$, $s \neq s'$, remain $O(1)$, then the uniformization rate $\Omega$ and the total number of state transitions
are $O(N)$. Thus, our algorithm now scales overall as $O(N^3)$, while the matrix exponentiation-based approach is $O(N^4)$. 
In either case, whether $A_{ss'}$ is $O(1)$ or $O(1/N)$, 
our algorithm is an order of magnitude faster.

\subsection{Bayesian Inference on the MJP Parameters} \label{sec:bayes_mjp}

In this section we briefly describe how full Bayesian analysis can be performed by placing priors on the MJP parameters $A$ and $\pi_0$ and sampling them as part of the MCMC algorithm. 
Like \cite{FearnSher2006}, we place independent gamma priors on the (negative) diagonal elements of $A$ and independent Dirichlet priors on the
transition probabilities. In particular, for all $s$ let $p_{s's} = {A_{s's}}/{|A_s|}$ and define the prior:
\begin{align}
 |A_s| & \sim \Gam(\alpha_1, \alpha_2) , \nonumber \\ %label{eq:A_prior1}\\
 (p_{s's}, s'\neq s) &\sim \Dir(\boldsymbol{\beta}) . \nonumber %\label{eq:A_prior2} %
\end{align}
This prior is conjugate, with sufficient statistics for the posterior distribution given a trajectory $\bS(t)$ being
the total amount of time $T_s$ spent in each state $s$ and the number of transitions $n_{s's}$ from each $s$ to $s'$.  In particular, 
\begin{align}
 |A_s|\, |\, (s_0,S,T) & \sim \Gam(\alpha_1+\textstyle \sum_{s'\neq s} n_{s's },\alpha_2+T_s) , \quad \text{and} \label{eq:A_post}\\
 (p_{s's}, s'\neq s)|(s_0,S,T)
    & \sim  \Dir(\boldsymbol{\beta}+ (n_{s's},s'\neq s)) \label{eq:p_post}%
\end{align}
It is important to note that we resample the rate matrix $A$ conditioned on $(s_0, S, T)$, and \emph{not} $(v_0, V, W)$. A new rate matrix $\tilde{A}$
implies a new uniformization rate $\tilde{\Omega}$, and in the latter case, we must also account for the probability of the Poisson events $W$ under 
$\tilde{\Omega}$. Besides being more complicated, this coupling between $W$ and $\Omega$ can slow down mixing of the MCMC sampler. Thus, we first discard the thinned events $U$, update $A$ conditioned only on the MJP trajectory, 
and then reintroduce the thinned events under the new parameters. We can view
the sampler of Algorithm \ref{alg:mcmc_smplr} as a transition kernel $\mathcal{K}_A((s_0,S,T),(\tilde{s}_0,\tS, \tT))$ that preserves the posterior distribution under the rate matrix $A$.
Our overall sampler then alternately updates $(s_0, S,T)$ via the transition kernel $\mathcal{K}_A(\cdot,\cdot)$, and 
then updates $A$ given $(s_0, S,T)$.

Finally, we can either fix $\pi_0$ or (as is sometimes appropriate) set it equal to the stationary distribution of the MJP with rate matrix $A$. In the latter case, Equations \eqref{eq:A_post} and \eqref{eq:p_post}
serve as a Metropolis-Hastings proposal. We accept a proposed $\tilde{A}$ sampled from this distribution with probability equal to the probability of the
current initial state under the stationary distribution of $\tilde{A}$. Note that computing this 
stationary distribution requires solving an $O(N^3)$ eigenvector problem, so that in this case, the overall Gibbs sampler scales cubically even though Algorithm \ref{alg:mcmc_smplr} scales quadratically.

\subsection{Experiments} \label{mjp:expt}

We first look at the effect of the parameter $\Omega$ on the mixing on the MCMC sampler. 
We generated a random 5-by-5 matrix $A$ (with hyperparameters $\alpha_1 = \alpha_2 = \beta = 1$), and used this to generate an MJP 
trajectory with a uniform initial distribution over states.
The state of this MJP trajectory was observed via a Poisson process likelihood model (see Section \ref{sec:mmpp}), and posterior samples given the observations and $A$ were produced by a C++ implementation of
our algorithm.  1000 MCMC runs were performed, each run consisting of $10000$ iterations after a burn-in of $1000$ iterations. 
 For each run, the number of transitions as well as the time spent in each state was calculated, and effective sample sizes (ESSs) of these statistics 
(the number of independent samples with the same `information' as the correlated MCMC samples) were calculated using R-CODA \citep{Rcoda2006}. 
The overall ESS of a run is defined to be the median ESS across all these ESSs. 
 
Figure \ref{fig:unifrate} (left) plots the overall ESS  against computation time per run, for different scalings $k$, where $\Omega=k\max_s|A_s|$.
We see that increasing $\Omega$ does increase
the mixing rate, however the added computational cost quickly swamps out any benefit this might afford. 
Figure \ref{fig:unifrate} (right) is a similar plot for the case where we also performed  Bayesian inference for the MJP parameter $A$ as described in Section \ref{sec:bayes_mjp}. 
Now we estimated the ESS of all off-diagonal elements of the matrix $A$, and the overall ESS of an
MCMC run is defined as the median ESS.  Interestingly, in this scenario, the ESS is fairly insensitive to 
$\Omega$, suggesting an `MCMC within Gibbs' update as proposed here using dependent trajectories is as effective as one using independent trajectories.  We found this to be true in general: when embedded within an outer 
MCMC sampler, our sampler produced similar effective ESSs as an MJP sampler that produces independent trajectories. The latter is typically more expensive,
and in any case, we will show that the computational savings provided by our sampler far outweigh the cost of dependent trajectories.

  \begin{figure}[ht]
  \begin{minipage}[b]{0.48\linewidth}
  \centering
    \includegraphics[width=\textwidth]{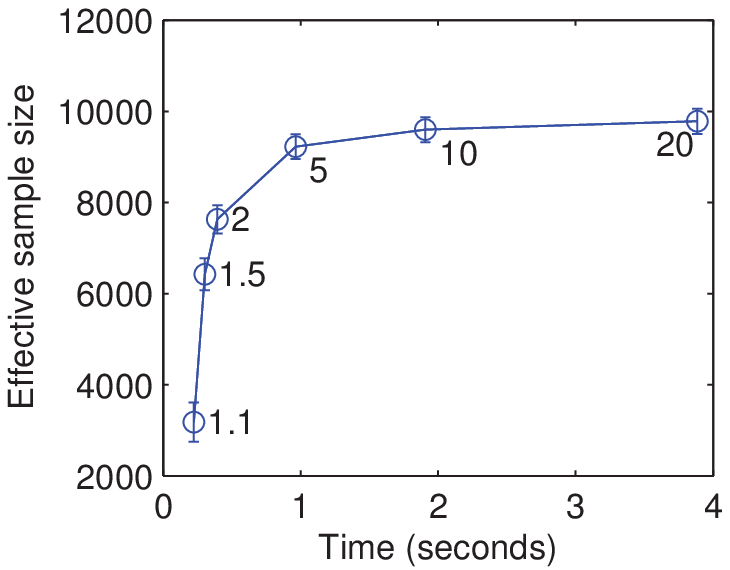} 
  \end{minipage}
  \begin{minipage}[b]{0.48\linewidth}
  \centering
    \includegraphics[width=\textwidth]{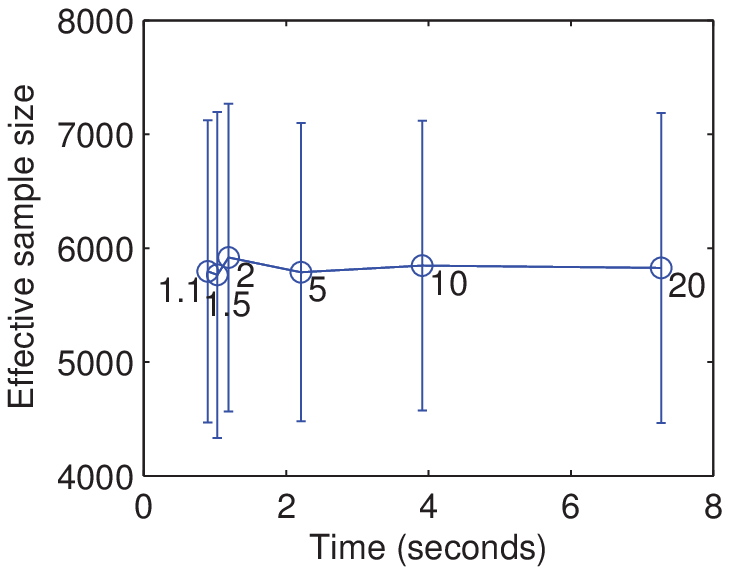} 
  \end{minipage}

  \caption{Effective sample sizes vs computation times for different scalings of $\Omega$ for (left) a fixed rate matrix $A$ and (right) Bayesian inference on $A$.  Whiskers are quartiles over 1000 runs.}

    \label{fig:unifrate}
  \end{figure}

In light of Figure \ref{fig:unifrate}, for all subsequent experiments we set $\Omega = 2 \max_s|A_s|$. Figure \ref{fig:burnin} shows the initial burn-in of a sampler with this setting for
different initializations.  The vertical axis shows the number of state transitions in the MJP trajectory of each iteration. This quantity quickly reaches its equilibrium value within a few iterations.

  \begin{figure}[ht]
  \centering
    \includegraphics[width=0.5\textwidth]{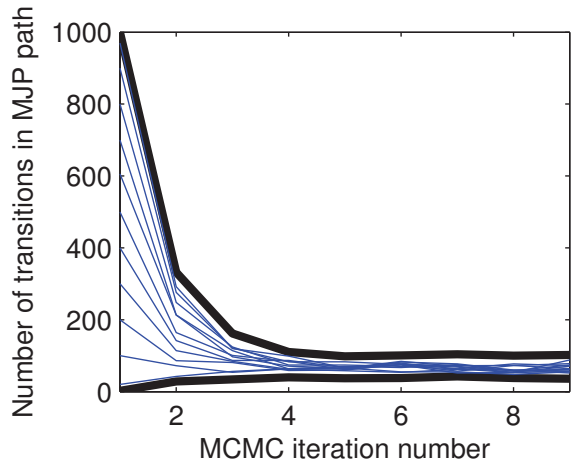} 
   \caption{Trace plot of the number of MJP transitions for different initializations.  Black lines are the maximum and minimum number of MJP transitions for each iteration, over all initializations.}
    \label{fig:burnin}
  \end{figure}

\section{Markov-Modulated Poisson Processes}\label{sec:mmpp}
A Markov modulated Poisson process  (MMPP) is a doubly-stochastic Poisson process whose intensity function is 
piecewise-constant and distributed according to a Markov jump process.
Suppose the MJP $(\bS(t),t\in[t_{start},t_{end}])$ has $N$ states, and is parametrized by an initial distribution over states $\pi_0$ and a rate matrix $A$. Associate with each state 
$s$ a nonnegative constant $\lambda_s$ called the emission rate of state $s$. Let $O$ be a set of points drawn from a Poisson process with piecewise-constant rate $R(t)=\lambda_{\bS(t)}$. 
Note that $O$ is unrelated to the subordinating Poisson process from the uniformization-based construction of the MJP, and we call
it the output Poisson process. 
The Poisson observations $O$ effectively form a continuous-time observation of the latent MJP, with the \emph{absence}
of Poisson events also informative about the MJP state. 
MMPPs have been used to model phenomenon like the distribution of rare DNA motifs along a gene \citep{FearnSher2006}, 
photon arrival in single molecule fluorescence experiments
\citep{Burzy2003}, and requests to web servers \citep{scottmmpp03}. 

\cite{FearnSher2006} developed an exact sampler for MMPPs based on a dynamic program for calculating the probability of $O$ marginalizing out the MJP trajectory.  
The dynamic program keeps track of the probability of the MMPP emitting all Poisson events prior to a time $t$ and ending in MJP state $s$.  The dynamic program then proceeds by 
iterating over all Poisson events in $O$ in increasing order, at each iteration updating probabilities using matrix exponentiation.  
A backward sampling step then draws an exact posterior sample of the MJP trajectory $(\bS(t), t\in O)$ evaluated at the times in $O$.  Finally a uniformization-based endpoint conditioned MJP sampler is used to fill in  the MJP trajectory between every pair of times in $O$.

The main advantage of this method is that it produces independent posterior samples. It does this at the price of being fairly complicated and computationally intensive. 
Moreover, it has the disadvantage of operating at the time scale of the Poisson observations rather than the dynamics of the latent MJP. For high Poisson rates, the 
number of matrix exponentiations will be high even if the underlying MJP has very low transition rates. This can lead to an inefficient algorithm.

Our MCMC sampler outlined in the previous section can be straightforwardly extended to the MMPP without any of these disadvantages. 
Resampling the auxiliary jump events (step $1$ in algorithm \ref{alg:blk_gibbs})
remains unaffected, since conditioned on the current MJP trajectory, they are independent of the observations $O$. Step $2$ requires calculating
the emission likelihoods $L_i(s)$, which is simply
given by:
\begin{align}
  L_i(s) &= \left(\lambda_{s}\right)^{|O_i|} \exp\left(-\lambda_{s}(w_{i+1} - w_i)\right) , \nonumber %\label{MMPP_lik}
\end{align}
 $|O_i|$ being the number of events of $O$ in the interval $[w_i,w_{i+1})$. %
Note that evaluating this likelihood only requires counting the number of observed Poisson events between every successive pair of times in $W$.  Compared to our algorithm, the approach of \cite{FearnSher2006} is much more involved and inefficient.

\subsection{Experiments} \label{mmpp_expts}
  In the following, we compare a C++ implementation of our algorithm with an implementation\footnote{Code was downloaded from Chris Sherlock's webpage.} of the algorithm of 
\cite{FearnSher2006}, coded in C. We performed fully Bayesian
inference, sampling both the MJP parameters (as described in Section \ref{sec:bayes_mjp}) and the Poisson rates $\lambda_s$ (conjugate gamma priors were
placed on these). In all instances, our algorithm did 
significantly better,
the performance improvement increasing with the complexity of the problem.

  \begin{figure}[ht]
  \begin{minipage}[b]{0.32\linewidth}
  \centering
    \includegraphics[width=\textwidth]{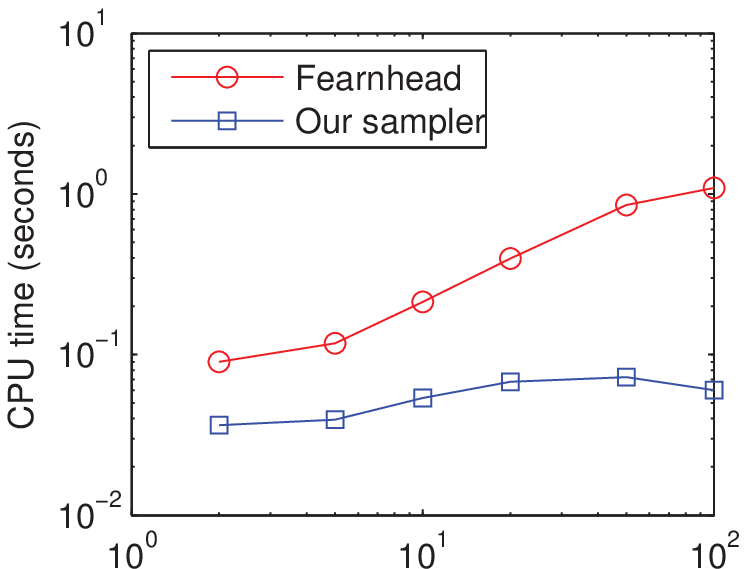} 
  \end{minipage}
  \begin{minipage}[b]{0.32\linewidth}
  \centering
    \includegraphics[width=\textwidth]{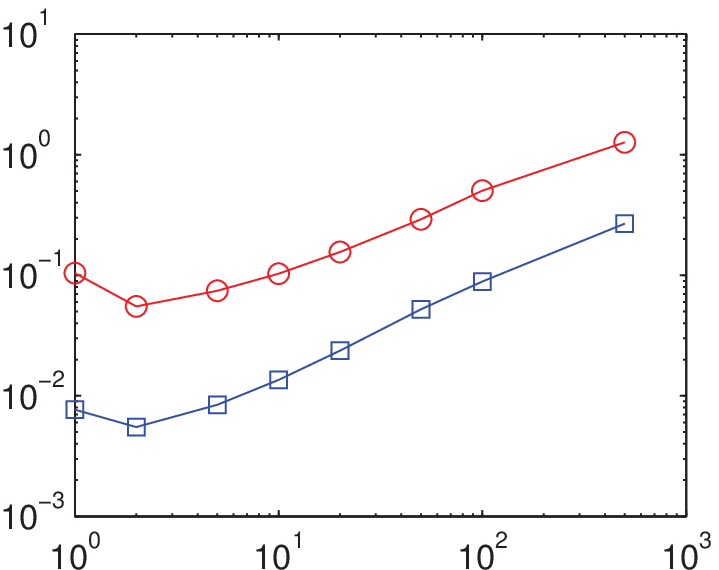} 
  \end{minipage}
  \begin{minipage}[b]{0.32\linewidth}
  \centering
    \includegraphics[width=\textwidth]{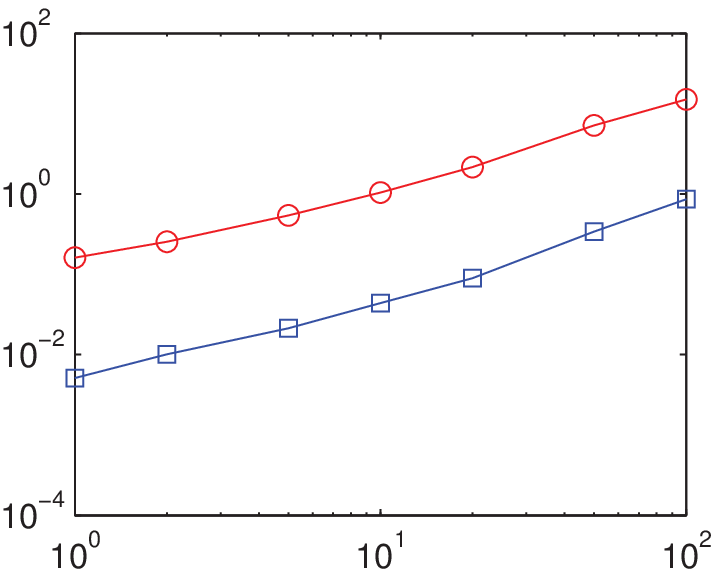} 
  \end{minipage}
    \caption{CPU time to produce 100 effective samples as we observe (left) increasing number of Poisson events in an interval of length 10, 
(centre) 10 Poisson events over increasing time intervals, and (right) increasing intervals with the number of events increasing on average.}
  \label{fig:mmpp_len}
  \end{figure}

In the first set of experiments, the dimension of the latent MJP was fixed to $5$.
The prior on the rate matrix $A$ had parameters $\alpha_1 = \alpha_2 = \beta = 1$ (see Section \ref{sec:bayes_mjp}). 
The shape parameter of the gamma prior on the emission rate of state $s$, $\lambda_s$, was set to $s$
(thereby breaking  symmetry across states); the scale parameter was fixed at $1$. 
10 draws of $O$ were simulated using the MMPP.  For each  observed $O$, both MCMC algorithms were run for 1000 burn-in iterations followed by 10000 iterations where samples were collected.    For each run, the ESS for each parameter was estimated using R-CODA, and the overall ESS was defined to be the median ESS over all parameters.  

Figure \ref{fig:mmpp_len} reports the average computation times required by each algorithm to produce 100 effective samples, under different scenarios.  The leftmost plot shows the computation times as a function of the numbers of Poisson events observed in an interval of fixed length $10$.  For our sampler, 
increasing the number of observed events leaves the computation time largely unaffected, while for the sampler of \cite{FearnSher2006}, this
increases quite significantly. This reiterates the point that our sampler works at the time scale of the latent MJP, while \cite{FearnSher2006} work at the 
time scale of the observed Poisson process.
  In the middle plot, we fix the number of observed Poisson events to $10$, while increasing the length of the observation interval instead, while in the rightmost plot, we 
increase both the interval length and the average number of observations in that interval. In both cases, our sampler again offers increased efficiency of
up to two orders of magnitude. In fact, the only problems where we observed the sampler of \cite{FearnSher2006} to outperform ours were low-dimensional
problems with only a few Poisson observations in a long interval, and with one very unstable state. %
A few very stable MJP states and a
few very unstable ones results in a high uniformization rate $\Omega$ but only a few state transitions. The resulting large number of virtual jumps can
make our sampler inefficient.

  \begin{figure}[ht]
  \centering
    \includegraphics[width=.5\textwidth]{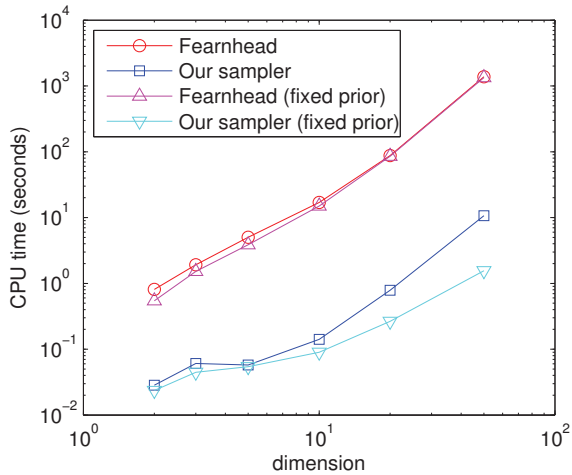} 
    \caption{CPU time to produce 100 effective samples as the MJP dimension increases}
    \label{fig:mmpp_dim}
  \end{figure}

In Figure \ref{fig:mmpp_dim}, we plot the time to produce $100$ effective samples as the number of states of the latent MJP increases. Here, we fixed the number
of Poisson observations to $10$ over an interval of length $10$. We see that our sampler (plotted with squares) offers substantial speed-up over the sampler of \cite{FearnSher2006}
(plotted with circles).  We see that for both samplers computation time scales cubically with the latent dimension. However, recall that this cubic scaling is not a
property of our MJP trajectory sampler; rather it is a consequence of using the equilibrium distribution of a sampled rate matrix as the initial distribution over states, which requires
calculating an eigenvector of a proposed rate matrix. If we fix the initial distribution over states (to the discrete uniform distribution), giving the line plotted with inverted triangles in the figure, we observe that our sampler scales quadratically.

\section{Continuous-Time Bayesian Networks (CTBNs)}\label{ctbn}

  \begin{figure}[!t]
  \begin{minipage}[b]{0.4\linewidth}
  \centering
    \includegraphics[width=.7\textwidth]{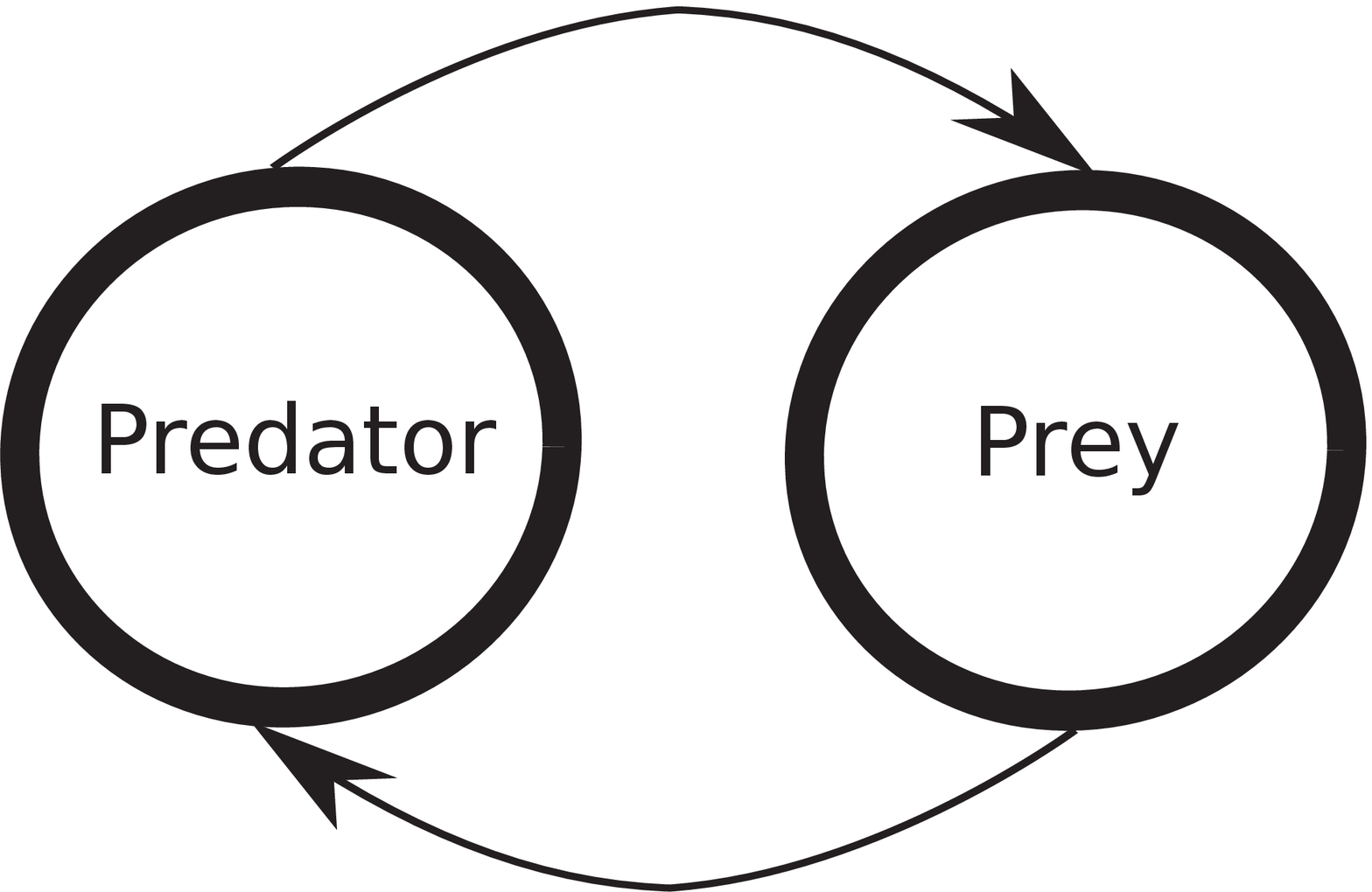} 
    \vspace{1in}
  \end{minipage}
  \begin{minipage}[b]{0.6\linewidth}
  \centering
    \includegraphics[width=.5\textwidth]{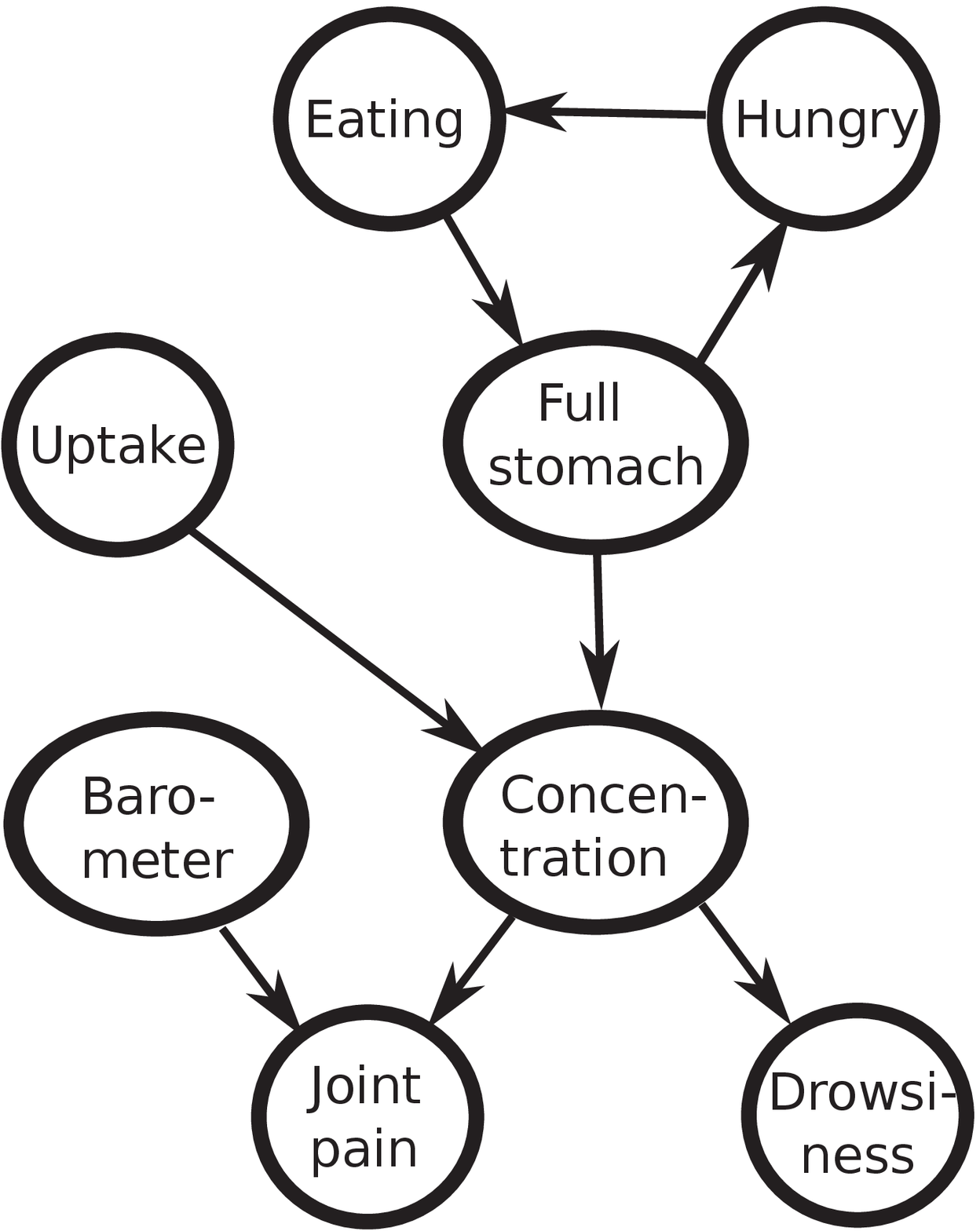} 
  \end{minipage}
   \caption[The predator-prey network  and the drug-effect CTBN]{The predator-prey network (left)  and the drug-effect CTBN (right)}
   \label{fig:drug_nw}
  \end{figure}

  Continuous-time Bayesian networks (CTBNs) are compact, multi-component representations of MJPs with structured rate matrices  \citep{Nodelman+al:UAI02}. Special instances of these models have 
long existed in the literature, particularly stochastic kinetic models like the Lotka-Volterra equations, which describe interacting populations of animal 
species, chemical reactants or gene regulatory networks \citep{wilk_stoch09}. 
There have also been a number of related developments, see for example \cite{Bolch1998} or \cite{didilez08}. For concreteness however, we shall focus on CTBNs, a formalism
introduced in \cite{Nodelman+al:UAI02} to harness the representational power of Bayesian networks to characterize structured MJPs. 

 Just as the familiar Bayesian network 
uses a product of conditional probability tables to represent a much larger probability table,
so too a CTBN represents a structured rate matrix with smaller conditional rate matrices. An $m$-component CTBN represents the state of an MJP at time $t$
with the states of $m$ nodes $\bS^1(t),\ldots,\bS^m(t)$ in a directed (and possibly cyclic) graph $\mathcal{G}$.
Figure \ref{fig:drug_nw} shows two CTBNs,
the `predator-prey network' and the `drug-effect network'. The former is a CTBN governed by the Lotka-Volterra 
equations (see subsection \ref{sec:lotka_volterra}), %
while the latter is used to model the dependencies in events leading to and following a patient taking a drug \citep{Nodelman+al:UAI02}. 

Intuitively, each node of the CTBN acts as an MJP with an instantaneous rate matrix that depends on the current configuration of its parents (and not its children, although the presence of directed cycles means a child can be a parent as well). 
The trajectories of all nodes are piecewise constant, and when a node changes state, the event rates of all its children change.
The graph $\mathcal{G}$  and the set of rate matrices (one for each node and for each configuration of its parents) characterize the dynamics
of the CTBN, the former describing the structure of the dependencies between various components, and the latter quantifying this.
Completing the specification of the CTBN is an initial distribution $\pi_0$ over the state of nodes, possibly specified via a Bayesian network. 

  \begin{figure}[!t]
  \centering
    \includegraphics[width=0.8\textwidth]{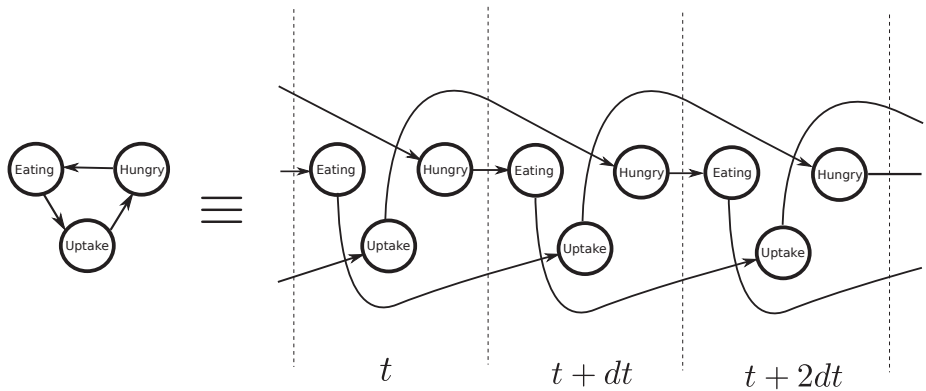} 
    \caption[The CTBN as a continuous time DBN]{Expanded CTBN}
    \label{fig:triang_exp}
\vspace{-.1in}
  \end{figure}

It is convenient to think of a CTBN as a compact representation of an expanded (and now acyclic) graph, consisting of
the nodes of $\mathcal{G}$ repeated infinitely along a continuum (viz.\ time). In this graph, arrows lead from a node at a time $t$ to instances of its 
children at time $t+\dif t$. 
Figure \ref{fig:triang_exp} displays this for a section of the drug-effect CTBN. 
The rates associated with a particular node at time $t+\dif t$ are determined by the configuration of its parents at time $t$. 
Figure \ref{fig:triang_exp} is the continuous-time limit of a class of discrete-time models called dynamic Bayesian networks (DBNs)
 \citep{Murphy02dynamicbayesian}. In a DBN, the state of a node at stage $i+1$ is dependent upon the configuration of its
parents at stage $i$.
Just as MJPs are continuous-time limits of discrete-time Markov chains, CTBNs are also continuous-time limits of DBNs.

It is possible to combine all local rate matrices of a CTBN into one global rate matrix \citep[see][]{Nodelman+al:UAI02}, 
resulting in a simple MJP whose state-space is the product state-space of all component nodes.
Consequently, it possible, conceptually at least, to directly sample a trajectory over an interval
$[t_{start},t_{end}]$ using Gillespie's algorithm. However, with an eye towards inference, 
Algorithm \ref{alg:ctbn_pr} 
describes a generative process that exploits the structure in the graph $\mathcal{G}$. 
Like Section \ref{sec:mjp}, we represent the trajectory of the CTBN, $\bS(t)$, with the initial state $s_0$ and the pair of sequences $(S,T)$. 
Let the CTBN have $m$ nodes. 
Now, $s_i$, the $i$th element of $S$, is an $m$-component vector representing the states of 
all nodes at $t_i$, the time of the $i$th state change of the CTBN.  We write this as $s_i = (s^1_i, \cdots, s^m_i)$. 
Let $k_i$ identify the component of the CTBN that changed state at $t_i$.
The rate matrix of a node $n$ varies over time as the
configuration of its parents changes, and we will write $A^{n,t}$ for the relevant matrix at time $t$.  Following Equation \eqref{eq:path_prob2}, we can write down the probability density of $(s_0, S,T)$ as
\begin{align}
p(s_0, S,T) &= \pi_0(s_0) \l( \prod_{i=1}^{|T|} A^{k_i,t_{i-1}}_{s^{k_i}_{i}s^{k_i}_{i-1}} \r)
\exp\l(-\sum_{k=1}^m \int_{t_{start}}^{t_{end}} |A^{k,t}_{\bS^k(t)}| \dif t\r) . \label{eq:ctbn_tr1} 
\end{align}

\begin{algorithm}[H]
\caption{Algorithm to sample a CTBN trajectory on the interval $[t_{start}, t_{end}]$}\label{alg:ctbn_pr}
\begin{tabular}{p{1.4cm}p{12.2cm}}
\textbf{Input:}  & The CTBN graph $\mathcal{G}$, a set of rate matrices $\{A\}$ for all nodes and for all \\
                 & parent configurations and an initial distribution
                   over states $\pi_0$.\\
\textbf{Output:} & A CTBN trajectory $\bS(t) \equiv (s_0, S,T)$.\\
\line(1,0){342}
\end{tabular}
\begin{algorithmic}[1]
\State Draw an initial configuration $s_0 \equiv (s^1_0,s^2_0,...) \sim \pi_0$. Set $t_0 = t_{start}$ and $i=0$. 
\Loop
\State For each node $k$, draw $z^k \sim \exp(|A^{k,t_i}_{s^k_i}|)$. \State Let $k_{i+1} = \argmin_k z^k$ be the first node to change state.
\State \textbf{If} $t_i+z^{k_{i+1}}>t_{end}$ \textbf{then return} $(s_0,\ldots,s_i,t_1,\ldots,t_i)$ and \textbf{stop}.
\State Increment $i$ and let $t_i = t_{i-1} + z^{k_i}$ be the next jump time. 
\State Let $s'=s^{k_i}_{i-1}$ be the previous state of node $k_i$.
\State Set $s^{k_i}_i = s$ with probability proportional to $ A^{k_i,t_{i-1}}_{ss'}$ for each $s\neq s'$.
\State Set $s^k_i = s^k_{i-1}$ for all $k \neq k_i$.
\EndLoop
\end{algorithmic}
\end{algorithm}

\subsection{Inference in CTBNs} \label{sec:ctbn_inf}
We now consider the problem of posterior inference over trajectories given some observations.
Write the parents and children of a node $k$ as $\mathcal{P}(k)$ and $\mathcal{C}(k)$ respectively.
Let $\mathcal{MB}(k)$ be the Markov blanket of node $k$, which consists of its parents, children, and the parents of its children.
Given the entire trajectories of all nodes in $\mathcal{MB}(k)$, node $k$ is independent of all other nodes in the network \citep{Nodelman+al:UAI02} (see also Equation \eqref{eq:ctbn_post} below). 
This suggests a Gibbs sampling scheme where the trajectory of each node is resampled given the configuration of its Markov blanket. This approach was followed 
by \cite{El_hay_gibbssampling}. 

However, even without any associated observations, sampling a node trajectory conditioned on the complete trajectory of its Markov 
blanket is not straightforward. 
To see this, rearrange the terms of Equation \eqref{eq:ctbn_tr1} to give
\begin{align}
p(s_0,S,T) &= 
 \pi_0(s_0) \prod_{k=1}^m \phi(S^k, T^k | s_0,S^{\mathcal{P}(k)}, T^{\mathcal{P}(k)}), \quad \text{and} \nonumber \\
 \phi(S^k, T^k | s_0,S^{\mathcal{P}(k)}, T^{\mathcal{P}(k)})
&=
\l(\prod_{i:k_i=k} A^{k,t_{i-1}}_{s^{k}_{i}s^{k}_{i-1}} \r)
\exp\l(-\int_{t_{start}}^{t_{end}} |A^{k,t}_{\bS^k(t)}| \dif t\r) ,  \label{eq:ctbn_cond1}
\end{align}
where for any set of nodes $B$, $(s_0^B,S^B,T^B)$ represents the associated trajectories.
Note that the $\phi(\cdot)$ terms are not conditional densities given parent trajectories, %
since the graph $\mathcal{G}$ can be cyclic. %
We must also account for the trajectories of node $k$'s children, so that the conditional density of $(s_0^k,S^k, T^k)$ is actually
\begin{align}
p(s^k_0,S^k, T^k| s^{\neg k}_0,S^{\neg k},T^{\neg k}) \propto &\pi_0(s_0^k|s_0^{\neg k}) \phi(S^k, T^k | s_0,S^{\mathcal{P}(k)}, T^{\mathcal{P}(k)}) \nonumber \\
&\cdot \prod_{c \in \mathcal{C}(k)} \phi(S^c, T^c | s_0,S^{\mathcal{P}(c)}, T^{\mathcal{P}(c)}) . \label{eq:ctbn_post}
\end{align}
Here $\neg k$ denotes all nodes other than $k$.
Thus, even over an interval of time where the parent configuration remains constant, the conditional distribution of 
the trajectory of a node is not a homogeneous MJP because of the effect of the node's children, which 
act as `observations' that are continuously observed.  
For any child $c$, if $A^{c,t}$ is constant over $t$, the corresponding $\phi(\cdot)$ is the density of an MJP given the initial state. 
Since $A^{c,t}$ varies in a piecewise-constant
manner according to the state of $k$, the $\phi(\cdot)$ term is actually the density of a piecewise-inhomogeneous MJP. 
Effectively, we have a `MJP-modulated MJP',
so that the inference problem here is a generalization of that for the MMPP of Section \ref{sec:mmpp}.

\cite{El_hay_gibbssampling} described a matrix-exponentiation-based algorithm to update the trajectory of node $k$. At a high-level their algorithm is similar to 
\cite{FearnSher2006} for MMPPs, with
the Poisson observations of the MMPP generalized to transitions in the trajectories of child nodes. Consequently, it uses an expensive forward-backward algorithm involving matrix exponentiations. In addition, \cite{El_hay_gibbssampling} resort to discretizing time via a binary search
to obtain the transition times upto machine accuracy.

\subsection{Auxiliary Variable Gibbs Sampling for CTBNs}

We now show how our uniformization-based sampler can easily be adapted to conditionally sample a trajectory for node $k$ without resorting to approximations.  
In the following, for notational simplicity.  we will drop the superscript $k$ whenever it is clear from context.  
For node $k$, the MJP trajectory $(s_0,S,T)$ has a uniformized construction from a subordinating Poisson process.  
The piecewise constant trajectories of the parents of $k$ imply that the MJP is piecewise homogeneous, and we will use a piecewise constant rate $\Omega^t$ which dominates the 
associated transition rates, i.e., $\Omega^t> |A^{k,t}_{s}|$ for all $s$.  This allows the dominating rate to `adapt' to the local transition rates, and 
is more efficient when, e.g., the transition rates associated with different parent configurations are markedly different.
Recall also that our algorithm first reconstructs the thinned Poisson events $U_{\cT}$ using a piecewise homogeneous Poisson process with rate 
$(\Omega^t + A^{k,t}_{\bS(t)})$, and then updates the trajectory using the forward-backward algorithm (so that
$W=T\cup U_{\cT}$ forms the candidate transitions times of the MJP).

In the present CTBN context, just as the subordinating Poisson process is inhomogeneous, so too the Markov chain used for the 
forward-backward algorithm will have different transition matrices at different times.  In particular, the transition matrix at a time $w_i$ (where $W=(w_1,\ldots,w_{|W|})$)  is 
\begin{align}
B_i=I+\frac{A^{k,w_i}}{\Omega^{w_i}} .  \nonumber
\end{align}

Finally, we need also to specify the likelihood function $L_i(s)$ accounting for the trajectories of the children in addition to actual observations in each time interval $[w_i,w_{i+1})$. 
From Equations \eqref{eq:ctbn_cond1} and \eqref{eq:ctbn_post}, this is given by
\begin{align}
  L_i(s) &= L_i^O(s) \prod_{c\in \mathcal{C}(k)}
\l(\prod_{j:k_j=k,t_j\in[w_i,w_{i+1})} \hspace{-.1in} A^{k,t_{j-1}}_{s^{k}_{j}s^{k}_{j-1}} \r)
\exp\l(-\hspace{-.02in} \int_{w_i}^{w_{i+1}} \hspace{-.05in} |A^{k,t}_{\bS^k(t)}| \dif t\r) , \nonumber
%\label{eq:ctbn_lik}
\end{align}
where $L_i^O(s)$ is the likelihood coming from actual observations dependent on the state of node $k$ in the time interval.  Note that the likelihood above depends only on the number of transitions each of the children make as well as how much time they spend in each state, for each parent configuration.

The new trajectory $\tbS^k(t)$ is now obtained using the forward-filtering backward-sampling algorithm, with the given inhomogeneous transition matrices and likelihood functions.  The following proposition now follows directly from our previous results in Section \ref{sec:MJP_MCMC_UNIF}:
\begin{prop}
The auxiliary variable Gibbs sampler described above converges to the posterior distribution over the CTBN sample paths. 
\end{prop}

Note that our algorithm produces a new trajectory that is dependent, through $T$, on the previous trajectory (unlike a true Gibbs update as in \cite{El_hay_gibbssampling} 
where they are independent).  However, we find that since the update is part of an overall Gibbs cycle over nodes of the CTBN, the mixing rate is actually dominated by dependence 
across nodes. Thus, a true Gibbs update has negligible benefit towards mixing, while being more expensive computationally.

\subsection{Experiments}\label{expr}
  In the following, we evaluate a C++ implementation of our algorithm on a number of CTBNs. As before, the dominating rate $\Omega^{t}$
was set to $ \max_s 2|A^{k,t}_s|$.%
\subsubsection{The Lotka-Volterra Process} \label{sec:lotka_volterra}

  We first apply our sampler to the Lotka-Volterra process \citep{wilk_stoch09, Opper_varinf}. Commonly referred to as the predator-prey model, this describes 
the evolution of two interacting populations of `prey' and `predator' species. The two species form the two nodes of a cyclic CTBN (Figure \ref{fig:drug_nw} (left)), whose states 
$x$ and $y$ represent the sizes of the prey and predator populations. The process rates are given by
\begin{align}
   A\l(\{x,y\} \rightarrow \{x+1,y\} \r)  &= \alpha x ,  &   
   A\l(\{x,y\} \rightarrow \{x-1,y\} \r)   &= \beta x y , \nonumber \\
   A\l(\{x,y\} \rightarrow \{x,y+1\} \r)  &= \delta x y , & 
   A\l(\{x,y\} \rightarrow \{x,y-1\} \r)   &= \gamma y , \nonumber
\end{align}
where we set the parameters as follows: 
$\alpha = 5 \times 10^{-4}, \beta = 1 \times 10^{-4}, \gamma = 5 \times 10^{-4}, \delta = 1 \times 10^{-4}$. All other rates are $0$.
  \begin{figure}[!t]
  \begin{minipage}[b]{0.5\linewidth}
  \centering
    \includegraphics[width=\textwidth]{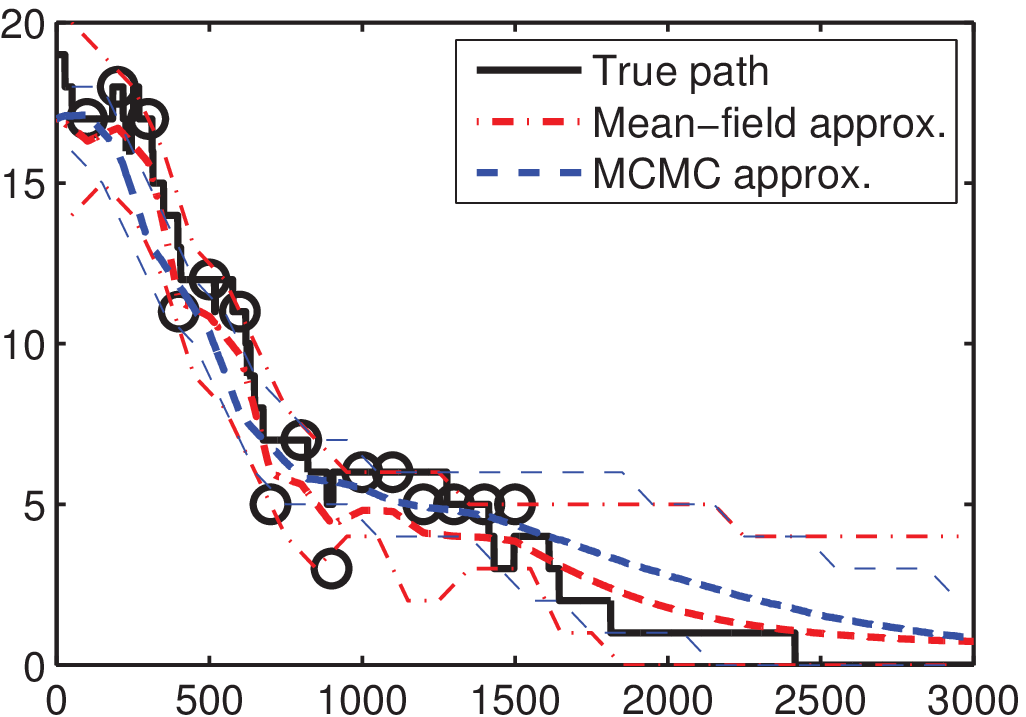} 
  \end{minipage}
  \begin{minipage}[b]{0.5\linewidth}
  \centering
    \includegraphics[width=\textwidth]{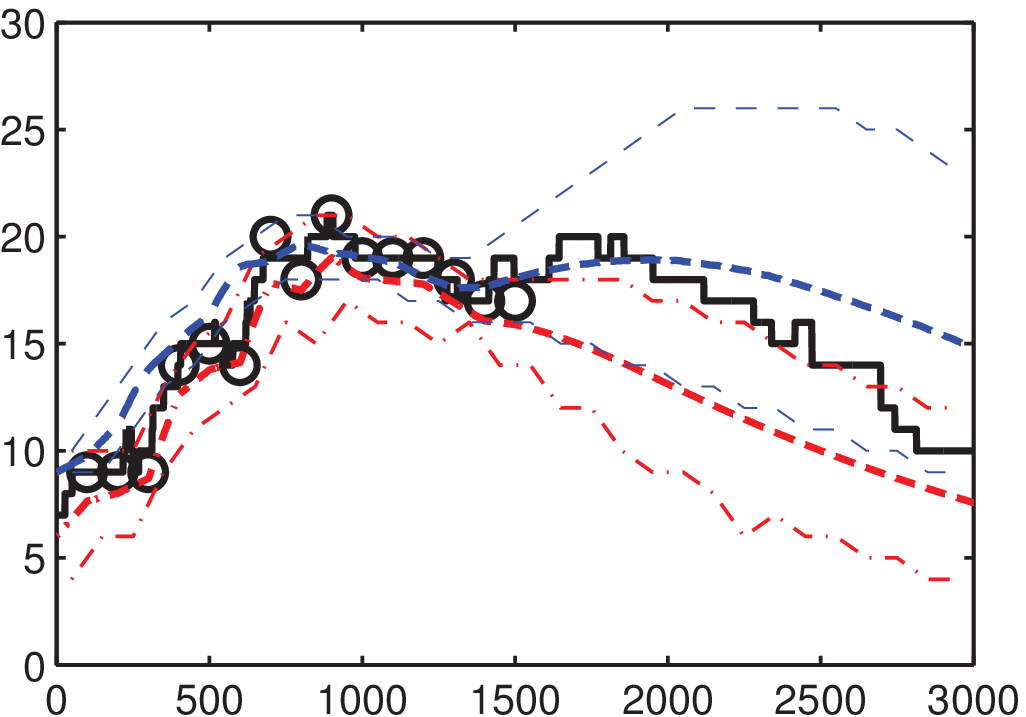} 
  \end{minipage}
   \caption[Posterior distributions over prey and predator populations]{Posterior (mean and 90\% credible intervals) over (left) prey and (right) predator paths (observations (circles) were available only until $1500$).}
   \label{fig:pred_prey}
  \end{figure}
This defines two infinite sets of infinite-dimensional conditional rate matrices. In its present form, our sampler cannot handle this infinite state-space 
\citep[but see][]{RaoTeh2012a}. 
Like 
\cite{Opper_varinf}, we limit the maximum number of individuals of each species to 200, 
leaving us with 400 rate matrices of size $200\times 200$. Note that these matrices are tridiagonal and very sparse: at any time the size of each population can change by 
at most one. Consequently, the complexity of our algorithm scales 
\emph{linearly} with the number of states (we did not modify our code to exploit this structure, though this is straightforward). A `true' path of predator-prey population
sizes was sampled from this process, and its state at time $t=0$ was observed noiselessly. Additionally 15 noisy observations were generated, spaced uniformly at intervals of $100$. The noise process was:
\begin{align}
  p(X(t)|\bS(t)) & \propto  \frac{1}{2^{|X(t)-\bS(t)|} + 10^{-6}} . \nonumber
\end{align}
Given these observations (as well as the true parameter values), we approximated the posterior distribution over paths by two methods: using 1000 samples from our 
MCMC sampler (with a burn-in period of 100) and using the mean-field (MF) approximation of \cite{Opper_varinf}\footnote{We thank Guido Sanguinetti for
providing us with his code.}. We could not apply the implementation of the Gibbs sampler of \citep{El_hay_gibbssampling} (see Section \ref{sec:avg_rle}) to a state-space and
time-interval this large.
Figure \ref{fig:pred_prey} shows the true paths (in black), the observations
(as circles) as well as the posterior means and $90\%$ credible intervals produced by the two algorithms for the prey (left) and predator (right) populations. As can be seen, 
both algorithms do well over the first half of the interval where data is present.  In the second half, the MF algorithm appears to underestimate the predicted size of the 
predator population.  On the other hand, the MCMC posterior reflects the true trajectory better. In general, we found the MF algorithm to underestimate the posterior variance in the MJP
trajectories, especially over regions with few observations.

\subsection{Average Relative Error vs Number of Samples}  \label{sec:avg_rle}
           
  \begin{figure}
  \centering
    \includegraphics[width=.5\textwidth]{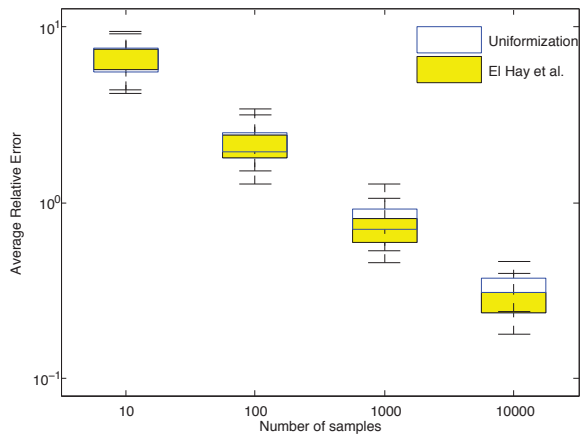} 
    \caption[Average relative error vs number of samples for CTBN samplers]{Average relative error vs number of samples for 1000 independent runs; burn-in = 200. Note that in this scenario, uniformization was about $12$
             times faster, so that for the same computational effort, it produces significantly lower errors.}
    \label{fig:mixing}
  \end{figure}

For the remaining experiments, we compared our sampler with the Gibbs sampler of \cite{El_hay_gibbssampling}.  For this comparison, we used the CTBN-RLE package of 
\cite{SheFanLamLeeXu10} (also implemented in C++). In all our experiments, as with the MMPP, we found our algorithm to be significantly faster, especially for larger
problems. To prevent details of the two implementations from clouding the picture and to reiterate the benefit afforded by avoiding matrix exponentiations, we also 
measured the amount of time CTBN-RLE spent exponentiating matrices. This constituted between $10\%$ to $70\%$ of the total running time of their
algorithm.  In the plots we refer to this as `El Hay et al. (Matrix Exp.)'. We found that our algorithm took less time than even this.

  In our first experiment,  we followed \cite{El_hay_gibbssampling} in studying how {average relative error} varies with the number of samples from the Markov chain. Average relative 
error is defined by $\sum_j \frac{|\hat{\theta_j} - \theta_j| } {\theta_j}$, and measures the total normalized difference between empirical ($\hat{\theta_j}$)
and true ($\theta_j$) averages of sufficient statistics of the posterior. The statistics in question are the time spent by each node in different 
states as well as the number of transitions from each state to the others. The exact values were calculated by numerical integration when possible, otherwise
from a very long run of CTBN-RLE.

 As in \cite{El_hay_gibbssampling}, we consider a CTBN with the topology of a chain, consisting of 5 nodes, each with 5 states. The states of the nodes were 
observed at times 0 and 20 and we produced endpoint-conditioned posterior samples of paths over the time interval $[0,20]$. We calculate the average relative error as a
function of the number of samples, with a burn-in of 200 samples. Figure \ref{fig:mixing} shows the results from running 1000 independent chains for both 
samplers.  Not surprisingly, the sampler of \cite{El_hay_gibbssampling}, which produces conditionally independent samples, has slightly lower errors. However
the difference in relative errors is minor, and is negligible when considering the dramatic (sometimes up to two orders of magnitude; see below) speed improvements of our algorithm. For instance, to produce the $10000$ samples, the \cite{El_hay_gibbssampling} sampler took about 6 minutes, while our sampler ran in about 30 seconds.

\subsection{Time Requirements for the Chain-Shaped CTBN}

  \begin{figure}[ht]
  \begin{minipage}[b]{0.32\linewidth}
  \centering
    \includegraphics[width=\textwidth]{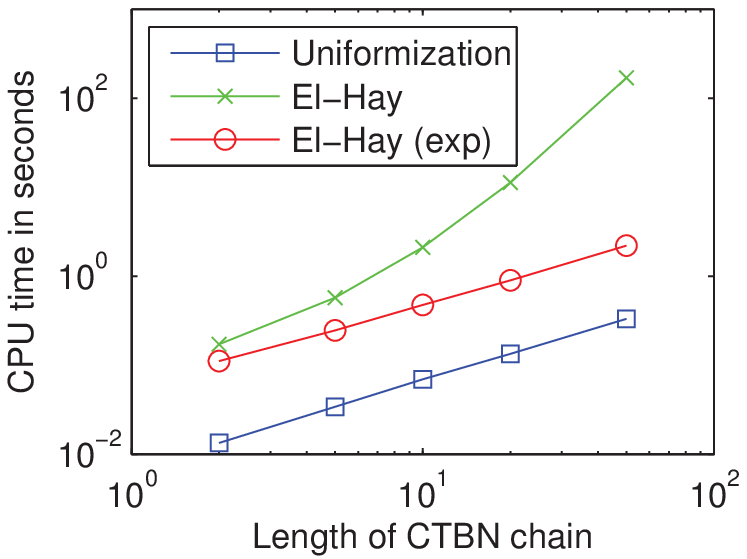} 
    \label{fig:len}
  \end{minipage}
  \begin{minipage}[b]{0.32\linewidth}
  \centering
    \includegraphics[width=\textwidth]{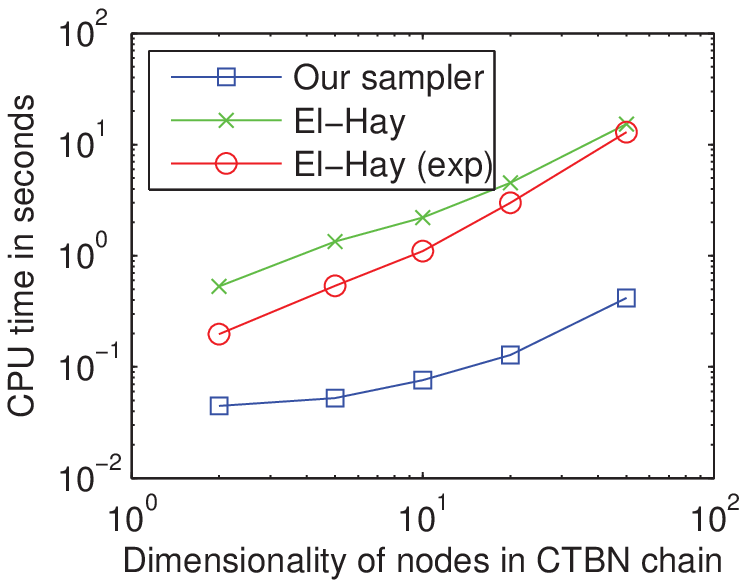} 
    \label{fig:dim}
  \end{minipage}
  \begin{minipage}[b]{0.32\linewidth}
  \centering
    \includegraphics[width=\textwidth]{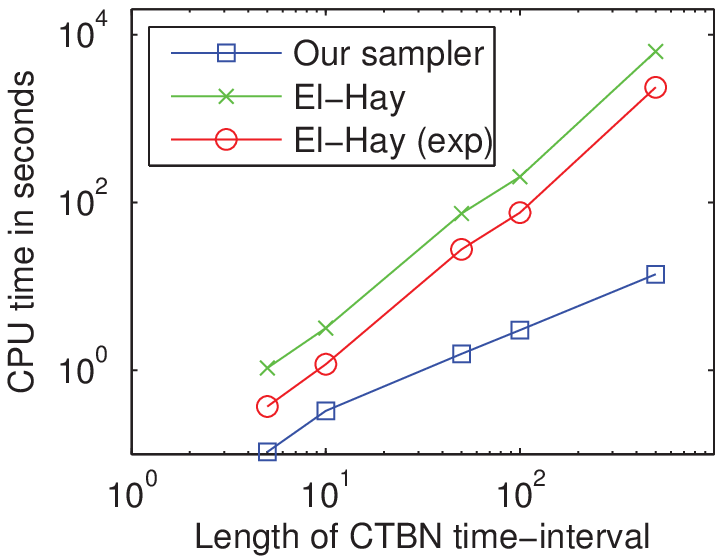} 
    \label{fig:int}
  \end{minipage}
  \caption{CPU time vs (left) length of CTBN chain (centre) number of states of CTBN nodes (right) time interval of CTBN paths.}
  \label{fig:len_dim_int}
  \end{figure}

In the next two experiments, we compare the times required by CTBN-RLE and our uniformization-based sampler to produce $100$ effective samples as the size of
the chain-shaped CTBN increased in different ways. In the first cases, we increased the length of the chain, and in the second, the dimensionality of each node.
In both cases, we produced posterior samples from an endpoint-conditioned 
CTBN with random gamma distributed parameters.

The time requirements were estimated from runs
of $10000$ samples after a burn-in period of $1000$ iterations. Since CTBN-RLE does not support Bayesian inference for CTBN parameters, we kept these fixed to
the truth.
To produce ESS estimates, we counted the number of transitions of each node and the amount of time spent in each state, and for each MCMC run, we
estimated the ESS of these quantities. Like in Section \ref{mmpp_expts}, the overall ESS is the median of these estimates. 
Each point in the figures is an average over $10$ simulations.

In the first of these experiments, we measured the times to produce $100$ effective samples for the chain-shaped CTBN described above, as the number of nodes in the chain 
(i.e., its length) increases. 
The leftmost plot in Figure \ref{fig:len_dim_int} shows the results. As might be expected, the time required by our algorithm grows linearly with the number of nodes. For
\cite{El_hay_gibbssampling}, the cost of the algorithm grows faster than linear, and quickly becoming unmanageable. The time spent calculating matrix exponentials 
\emph{does} grow linearly, however our uniformization-based sampler always takes less time than even this.

Next, we kept the length of the chain fixed at $5$, instead increasing the number of states per node.  As seen in the middle plot, once again, our sampler is always faster. Asymptotically,
we expect our sampler to scale as $O(N^2)$ and \cite{El_hay_gibbssampling} as $O(N^3)$.  While we have not hit that regime yet, we can see that the cost of our sampler
grows more slowly with the number of states.

\subsection{Time Requirements for the Drug-Effect CTBN}
Our final experiment, reported in the rightmost plot of Figure \ref{fig:len_dim_int}, measures the time required as the interval length $(t_{end} - t_{start})$ increases. For this experiment, we used
the drug-effect network shown in Figure \ref{fig:drug_nw}, where the parameters were set to standard values (obtained from CTBN-RLE) and the state of the network was fully observed at
the beginning and end times. Again, our algorithm is the faster of the two, showing a linear increases in computational costs with the length of the interval. 
It is worth pointing
out here that the algorithm of \cite{El_hay_gibbssampling} has a `precision' parameter, and that by reducing the desired temporal precision, faster performance
can be obtained. However, since our sampler produces \emph{exact} samples (up to numerical precision), our comparison is fair. In the above experiments,
we left this parameter at its default value.

\vspace{-.1in}

\section{Discussion}\label{discussion}

We proposed a novel Markov chain Monte Carlo sampling method for Markov jump processes. Our method exploits the simplification of the structure of the MJP resulting from the 
introduction of auxiliary variables via the idea of uniformization.  This constructs a Markov jump process by subordinating a Markov chain to a Poisson process, and 
amounts to running a Markov chain on a random discretization of time. Our sampler 
is a blocked Gibbs sampler in this augmented representation and proceeds by alternately resampling the discretization given the Markov chain and vice versa. Experimentally, we find that this auxiliary variable Gibbs 
sampler is computationally very efficient. %
The sampler easily generalizes to other MJP-based models, and we presented samplers for Markov-modulated Poisson processes and continuous-time Bayesian networks. In our experiments, we showed significant speed-up compared to state-of-the-art samplers for both.

Our method opens a number of avenues worth exploring. One concerns the subordinating Poisson rate $\Omega$ which acts as a free-parameter of the sampler. While our heuristic
of setting this to $\max_s 2|A_s|$ worked well in our experiments, this may not be the case for rate matrices with widely varying transition rates. A possible approach is to `learn' a good
setting of this parameter via adaptive MCMC methods. More fundamentally, it would be interesting to investigate if theoretical claims can be made about the `best' setting
of this parameter under some measures of mixing speed and computational cost.

Next, there are a number of immediate generalizations of our sampler. First, our algorithm is easily applicable to inhomogeneous Markov jump processes where 
techniques based on matrix exponentiation cannot be applied. Following recent work \citep{RaoTeh2011b}, we can also look at generalizing our sampler to semi-Markov processes 
where the holding times of the states follow non-exponential distributions.  These models find applications in fields like biostatistics, neuroscience and queuing theory \citep{mode1988}.
By combining our technique with slice sampling ideas \citep{Nea2003a}, we can explore Markov jump processes with countably infinite state spaces. Another generalization 
concerns MJPs with unbounded rate matrices. For the predator-prey model, we avoided this problem by bounding the maximum population sizes; otherwise %
it is impossible to choose a dominating $\Omega$. Of course, in practical settings, any trajectory from this process is bounded with 
probability $1$, and we can extend our method to this case by treating $\Omega$ as a trajectory dependent random variable. 
For some work in this direction, we refer to \cite{RaoTeh2012a}.

\vspace{-.1in}
\acks{This work was done while both authors were at the Gatsby Computational Neuroscience Unit, University College London.
We would like to acknowledge the Gatsby Charitable Foundation for generous funding. We thank the associate editor and the anonymous reviewers for useful
comments.}

\appendix
\section{The Forward-Filtering Backward-Sampling (FFBS) Algorithm} \label{sec:ffbs}

For completeness, we include a description of the forward-filtering backward-sampling algorithm for discrete-time Markov chains. The earliest references for
this that we are aware of are \cite{fru:StateSpace} and \cite{Carter96}.

Let $S_t,\ t \in \{0,\cdots T\}$ be a discrete-time Markov chain with a discrete state space $\cS \equiv \{1,\cdots N\}$.
We allow the chain to be inhomogeneous, with $B^t$ being the state transition matrix at time $t$ (so that $p(S_{t+1} = s'|S_t = s) = B^t_{s's}$). Let $\pi_0$ be
the initial distribution over states at $t=0$. Let $O^t$ be a noisy observation of the state at time $t$, with the likelihood given by $L^t(s) = p(O^t|S^t = s)$.
Given a set of observations $\mathcal{O} = (O^0,\cdots,O^{T})$, FFBS returns an independent posterior sample of the state vector.

Define $\alpha^t(s) = p(O^0,\cdots, O^{t-1},S^{t} = s)$. From the Markov property, we have the following recursion:
\begin{align}
  \alpha^{t+1}(s') = \sum_{s=1}^N \alpha^t(s) L^t(s) B^t_{s's}. \nonumber 
\end{align}

Calculating this for all $N$ values of $s'$ takes $O(N^2)$ computation, and a forward pass through all $T$ times is $O(TN^2)$.
At the end of the forward pass, we have a vector 
\begin{align}
 \beta^T(s) := L^T(s) \alpha^T(s) &= p(\mathcal{O}, S_T = s) \propto p(S_T = s | \mathcal{O}) . \nonumber 
\end{align}
{It is easy to sample a realization of $S_T$ from this. Next, note that}
\begin{align}
 p(S_t = s |S_{t+1} = s', \mathcal{O}) &\propto  p(S_t = s, S_{t+1} = s', \mathcal{O}) \nonumber \\
%   &=  p(O^1, \cdots, O^t, S_t = s, O^{t+1}, \cdots, O^T, S_{t+1} = s', \mathcal{O}) \\
   &= \alpha^t(s) B^t_{s's} L^t(s) p(O^{t+1}, \cdots, O^T| S_{t+1} = s') \nonumber \\
   &\propto \alpha^t(s) B^t_{s's} L^t(s) , \nonumber 
%   & := \beta^t(s)
\end{align}
where the second equality follows from the Markov property. This too is easy to sample from, and the backward pass of FFBS successively samples $S_{T-1}$ to
$S_0$. We thus have a sample $(S_0,\cdots, S_T)$. The overall algorithm is given below:

\begin{algorithm}[H]
\caption{The forward-filtering backward-sampling algorithm}\label{alg:ffbs}
\begin{tabular}{p{1.4cm}p{10.5cm}}
\textbf{Input:}  & An initial distribution over states $\pi_0$, a sequence of transition matrices $B^t$,
                   a sequence of observations $\mathcal{O} = (O_1, \cdots O_T)$ with likelihoods $L^t(s) = p(O^t|S^t = s)$.\\

\textbf{Output:} &A realization of the Markov chain $(S_0, \cdots, S_{T})$.\\
\hline 
\end{tabular}
\begin{algorithmic}[1]
\State Set $\alpha^0(s) = \pi_0(s)$.
\For{$t = 1 \to T$}
\State  $\alpha^{t}(s') = \sum_{s=1}^N \left( \alpha^{t-1}(s) L^{t-1}(s) B^{t-1}_{s's} \right) \quad \text{for } s' \in \{1,\cdots, N\}$.
\EndFor 
\State Sample $S_T \sim \beta^T(\cdot)$, where $\beta^T(s) := L^T(s) \alpha^T(s)$. 
\For{$t = T \to 0$}
  \State Define $\beta^t(s) = \alpha^t(s)B^t_{S_{t+1}s} L^t(s)$.
  \State Sample $S_t \sim \beta^t(\cdot)$.
\EndFor 
\State \textbf{return} {$(S_0,\cdots,S_T)$}.
\end{algorithmic}
\end{algorithm}

\bibliography{./refvr}

\end{document}